\documentclass[12pt]{amsart}

\usepackage{amsmath,amssymb,amsthm,graphicx}
\usepackage{verbatim}

\newcommand{\CC}{\mathbb{C}}

\newcommand{\RR}{\mathbb{R}} 
\newcommand{\RRp}{\RR_+}
\renewcommand{\SS}{\mathbb{S}}

\newcommand{\ZZ}{\mathbb{Z}}

\newcommand{\cB}{\mathcal{B}} 
\newcommand{\cI}{\mathcal{I}}

\newcommand{\cF}{\mathcal{F}}
\newcommand{\cE}{\mathcal{E}}

\newcommand{\cH}{\mathcal{H}}
\newcommand{\cO}{\mathcal{O}}

\newcommand{\cQ}{\mathcal{Q}}
\newcommand{\oD}{D}
\newcommand{\tc}{\mathrm{c}}
\newcommand{\oper}{\nu}
\newcommand{\wt}{w}
\newcommand{\Qab}{\cQ^{A,B}}
\newcommand{\Qabxy}{\cQ^{A\sd xy,B\sd xy}}

\newcommand{\fr}{\phi_{\rho}}
\newcommand{\frr}{\phi_{\rho'}}

\renewcommand{\a}{\alpha}
\newcommand{\D}{\Delta} 
\renewcommand{\d}{\delta} 
\newcommand{\G}{\Gamma}
\newcommand{\Gh}{\G}
\newcommand{\g}{\gamma} 
\renewcommand{\L}{\Lambda} 
\renewcommand{\l}{\lambda}
\renewcommand{\b}{\beta}

\renewcommand{\S}{\Sigma} 
\newcommand{\s}{\sigma}
  
\newcommand{\eps}{\varepsilon}
 
\renewcommand{\max}{\mathrm{max}}

\newcommand{\el}{\langle} 
\newcommand{\er}{\rangle}
\newcommand{\tr}{\mathrm{tr}}
\newcommand{\ev}{\mathrm{ev}}
\newcommand{\odd}{\mathrm{odd}}
\newcommand{\rf}{\mathrm{f}}
\newcommand{\sd}{\bigtriangleup}

\newcommand{\fk}{{\sc fk}}
\newcommand{\fkg}{{\sc fkg}}
\newcommand{\gks}{{\sc gks}}
\newcommand{\ghs}{{\sc ghs}}
\newcommand{\lra}{\leftrightarrow}
\newcommand{\nlra}{\nleftrightarrow}

\newcommand{\bc}{\rho_{\mathrm{c}}}
\newcommand{\bs}{\rho_{\mathrm{s}}}
\renewcommand{\b}{\beta}
\newcommand{\oo}{\infty}
\newcommand{\qq}{\quad\quad}
\newcommand{\rc}{random-cluster}
\newcommand{\bra}[1]{\langle#1|}
\newcommand{\ket}[1]{|#1\rangle}
\newcommand{\Si}{\Sigma}
\newcommand{\sm}{\setminus}
\newcommand{\resp}{respectively}
\newcommand{\even}{\mathrm{even}}
\renewcommand{\odd}{\mathrm{odd}}
\renewcommand{\o}{\mathrm{o}}
\newcommand{\wtilde}{\widetilde}
\newcommand{\what}{\widehat}
\newcommand{\pd}{\partial}
\newcommand{\es}{\varnothing}
\newcommand{\od}{d}
\newcommand{\ol}{\overline}
\newcommand{\lrao}[1]{\overset{#1}{\lra}}
\newcommand{\bigmid}{\,\big|\,}
\newcommand{\Bigmid}{\,\Big|\,}

\theoremstyle{definition} 
\newtheorem{definition}{Definition}[section]
\newtheorem{theorem}[definition]{Theorem}

\newtheorem{lemma}[definition]{Lemma}

\newtheorem{remark}[definition]{Remark}

\newtheorem{assumption}[definition]{Assumption}
\newcommand\urladdrx[1]{{\urladdr{\def~{{\tiny$\sim$}}#1}}}
\newcounter{mycount}
\newenvironment{romlist}{\begin{list}{\rm(\roman{mycount})}%
   {\usecounter{mycount}\labelwidth=1cm\itemsep 0pt}}{\end{list}}

\newenvironment{letlist}{\begin{list}{(\alph{mycount})}%
   {\usecounter{mycount}\labelwidth=1cm\itemsep 0pt}}{\end{list}}
\numberwithin{equation}{section}

\begin{document}
\title[Phase transition of the quantum Ising model]
{The phase transition of the\\quantum Ising model is sharp}

\author{J. E. Bj\"ornberg}
\author{G. R. Grimmett}
\address{Statistical Laboratory,
Centre for Mathematical Sciences,
University of Cambridge,
Wilberforce Road, Cambridge CB3 0WB, U.K.}
\email{jeb76@cam.ac.uk}
\email{g.r.grimmett@statslab.cam.ac.uk}
\urladdrx{http://www.statslab.cam.ac.uk/~grg/}

\date{2 January 2009, revised 20 June 2009}

\begin{abstract}
An analysis is presented of the phase transition of the quantum Ising
model with transverse field on the $d$-dimensional hypercubic lattice.
It is shown that there is a unique sharp transition.
The value of the critical point is calculated rigorously in one dimension.
The first step is to express the quantum Ising model in terms of a 
(continuous) classical
Ising model in $d+1$ dimensions. A so-called `random-parity'
representation is developed for the latter model, similar to 
the random-current representation for the classical Ising model on a discrete lattice.
Certain differential inequalities are proved. 
Integration of these inequalities yields the
sharpness of the phase transition, and also a number of other facts concerning
the critical and near-critical behaviour of the model under study.
\end{abstract}


\maketitle

\section{Introduction}\label{sec-intro}
Geometric or `graphical' methods have been very useful in the rigorous study of
lattice models in classical statistical mechanics.  Of the many examples, we mention the
use of the random-cluster (or `\fk') representation to prove the existence of 
non-translation-invariant `Dobrushin' states in the $q$-state Potts
model~\cite{gielis_grimmett};  the use of the related `loop' representation
to prove conformal invariance for the two-dimensional Ising
model~\cite{smirnov_ising};  the use of the random-current 
representation to prove the sharpness of the phase transition in
classical Ising models~\cite{abf}.  In contrast, 
graphical methods for \emph{quantum} lattice models have 
received less attention.  
We shall formulate a so-called `random-parity representation' for the quantum
Ising model on a graph $G$ (or, more precisely, for the corresponding
`continuous Ising model' on $G \times \RR$, \cite{akn,aizenman_nacht}), 
and shall use it to prove the sharpness of the phase
transition for this model in a general number of dimensions. 
The random-parity representation is a cousin of the random-current
representation in \cite{aiz82,abf}.

Let $L=(V,E)$ be a finite graph.
The Hamiltonian of the quantum Ising
model with transverse field on $L$ is the matrix (or `operator')
\begin{equation}\label{qi_ham_eq}
H=-\tfrac{1}{2}\l\sum_{e=uv\in E}\s_u^{(3)}\s_v^{(3)}-\d\sum_{v\in V}\s_v^{(1)},
\end{equation}
acting on the Hilbert space $\cH=\bigotimes_{v\in V}\CC^2$.  Here, the Pauli
spin-$\frac12$ matrices are given as 
\begin{equation}
\s_v^{(3)}=
\begin{pmatrix} 
1 & 0 \\
0 & -1
\end{pmatrix},\qquad
\s_v^{(1)}=
\begin{pmatrix} 
0 & 1 \\
1 & 0
\end{pmatrix}.
\end{equation}
The constants $\l,\d>0$ in \eqref{qi_ham_eq} are the spin-coupling and transverse-field intensities,
respectively.   The basic operator of the quantum Ising model is $e^{-\b H}$ where
$\b>0$. The model was introduced in \cite{lieb},
and has been widely studied since. See, for example, the references in \cite{GOS}.

It is standard (see \cite{akn,aizenman_nacht} for example) 
that the quantum Ising model on $L$ possesses
a  type of `path integral representation', which expresses it as a 
type of classical Ising model (or equivalently as
a continuum \rc\ model with $q=2$) on the continuous space $V \times [0,\b]$.  
This representation permits the use of geometrical methods in studying
the behaviour of the original quantum model. In particular, it
is a useful way of establishing the existence of the infinite-volume limits as $\b\to\oo$ 
and $|V|\to\oo$, and of relating the phase transition of the quantum model
to that of the continuous classical model. 
 
The main technique of this article is a type of random-current representation,
called the `random-parity' representation, for the 
Ising model on $V\times \RR$. This enables a detailed analysis of the phase
transition of the latter model, and hence of the related quantum model. 
Further details and
references will be provided in the next section.  

The quantum model is said to be in the `ground state' when the limit $\b\to\oo$ is taken.
The value of $\b$ appears in the superscript of quantities that follow; when 
the superscript is $\oo$, this is to be interpreted as the relevant ground-state quantity. 

Our main choice for $L$ is a box in the $d$-dimensional cubic
lattice $\ZZ^d$ where $d \ge 1$, with a periodic boundary condition, and we shall
pass to the infinite-volume limit as $L\uparrow \ZZ^d$. (Similar results
hold for other lattices, and for summable translation-invariant interactions.)
The model is
over-parametrized. We shall normally assume $\d = 1$, and write $\rho=\l/\d$, while
noting that the same analysis holds for $\d\in(0,\oo)$.
As remarked above,
one may study the quantum phase transition via that of the Ising model
on the continuum $\ZZ^d \times [0,\b]$ and, in the latter case, one may introduce the notions of magnetization
$M=M^\b(\rho,\g)$ and (magnetic) susceptibility $\chi=\chi^\b(\rho,\g)$,
where $\g$ denotes external field. The
\emph{critical point} $\bc=\bc^\b$ is given by
\begin{equation}
\bc^\b:=\inf\{\rho: M^\b_+(\rho)>0\},
\label{o1}
\end{equation}
where
\begin{equation}
M^\b_+(\rho):=\lim_{\g\downarrow 0}M^\b(\rho,\g),
\label{o2}
\end{equation}
is the magnetization in the limiting state $\el\cdot\er^\b_+$ as
$\g\downarrow 0$.
It may be proved by standard methods that: 
\begin{equation}
\begin{split}
\text{if $d \ge 2$}&:\quad 0<\bc^\b<\oo\text{ for $\b\in(0,\oo]$},\\ 
\text{if $d=1$}&:\quad \bc^\b=\oo \text{ for $\b\in(0,\oo)$}, \ 0<\bc^\oo<\oo. 
\end{split}
\label{critvals}
\end{equation}
When $\b<\oo$, the magnetization, susceptibility, and critical values depend also
on the parameter $\l$, but we suppress this for brevity of notation.

Complete statements of our main results are deferred until Sections \ref{cons_sec}
and \ref{sec_1d}.
Here are two examples of what can be proved.

\begin{theorem}\label{ed_thm}
Let $u,v\in \ZZ^d$ where $d \ge 1$, and $s,t\in\RR$. For $\b\in(0,\oo]$:
\begin{romlist}
\item
if $0 < \rho < \bc^\b$, the two-point correlation function 
$\langle \s_{(u,s)}\s_{(v,t)}\rangle^\b_+$ of the
Ising model on $\ZZ^d\times \RR$ decays exponentially to $0$ 
as $|u-v|+|s-t|\to\oo$,
\item
if $\rho > \bc^\b$,
 $\langle \s_{(u,s)}\s_{(v,t)}\rangle^\b_+ \ge M^\b_+(\rho)^2>0$.
\end{romlist}
\end{theorem}

\begin{theorem}\label{mf_thm} Let $\b\in(0,\oo]$. In the notation of Theorem \ref{ed_thm},
there exists $c = c(d)>0$ such that 
$$
M^\b_+(\rho) \geq c(\rho-\bc^\b)^{1/2}\qq\text{for }
\rho>\bc^\b.
$$ 
\end{theorem}

These and other facts will be stated and proved in Section \ref{cons_sec}.
Their implications for the infinite-volume quantum model will be elaborated in the next section, 
see in particular \eqref{o4}--\eqref{o5}. Roughly speaking, they imply that
the two-point function of the quantum model decays exponentially when $\rho<\bc^\b$,
and is uniformly bounded below by $c(\rho-\bc^\b)$ when $\rho>\bc^\b$.

The approach used here is to prove a family of differential
inequalities for the magnetization $M^\b(\rho,\g)$. This parallels the methods
established in \cite{ab,abf} for the analysis of the phase transitions in
percolation and Ising models on discrete lattices, and indeed our arguments are closely
related to those of \cite{abf}. Whereas new problems arise in
the current context and require treatment, certain aspects of the analysis
presented here are simpler than the corresponding steps of \cite{abf}. 
The application to the quantum model imposes a periodic boundary condition in
the $\b$ direction; the same conclusions are valid
for the space--time Ising model with a free boundary condition.

The critical value $\bc^\b$ depends of course on the number of dimensions.
We shall use planar duality to show that $\bc^\oo=2$ when $d=1$,
and in addition that the transition is of second order in that $M^\oo_+(2)=0$. See Theorem
\ref{crit_val_cor}. The one-dimensional
critical point has been calculated by other means in the quantum case,
but we believe that the current proof is valuable. Two applications
to the work of \cite{bjo0,GOS} are summarized in Section \ref{sec_1d}.

Here is a brief outline of the contents of this article.  Formal definitions are presented
in Section \ref{sec-backg}. The random-parity 
representation of the quantum Ising model is described in Section \ref{rcr_sec}.
This representation may at first sight seem quite different from 
the random-current representation of the classical Ising model on a discrete lattice.
It requires more work to set up than does its discrete cousin, 
but once in place it works in a very similar, and sometimes simpler, manner.  
We then state and prove, in Section~\ref{ssec-switching}, the fundamental
`switching lemma'.  In
Section~\ref{sw_appl_sec} are presented a number of important consequences
of the switching lemma, including {\ghs} and
Simon--Lieb 
inequalities, as well as other useful inequalities and identities.
In Section~\ref{pf_sec}, we prove the somewhat more involved differential inequality
of the forthcoming Theorem \ref{main_pdi_thm}, 
which is similar to the main inequality of \cite{abf}. Our main results  
follow from Theorem \ref{main_pdi_thm} in conjunction with the results of
Section \ref{sw_appl_sec}.
Finally, in Sections \ref{cons_sec} and \ref{sec_1d}, we give rigorous 
formulations and proofs of our main results.

We mention that the continuous Ising model possesses a representation
of \rc-type; see, for example, \cite{akn,grimmett_stp,GOS}. This is convenient
for proving various facts including the existence of infinite-volume limits.
Only occasional use is made of the \rc\ representation here, and full details are omitted.
See Remark \ref{rc_unique}.

\begin{remark}\label{crawford}
There is a very substantial overlap between the results reported
here and those of the independent and contemporaneous article
\cite{craw-i}. The basic differential inequalities of Theorems
\ref{main_pdi_thm} and \ref{three_ineq_lem} appear in both articles. The proofs are 
in essence the same despite some differences of presentation. We are grateful to
the authors of \cite{craw-i} for explaining the relationship between the random-parity
representation of Section \ref{rcr_sec} and the random-current 
representation of \cite[Sect.\ 2.2]{ioffe_geom}. As pointed out in \cite{craw-i}, 
the appendix of \cite{chayes_ioffe_curie-weiss} contains a type of switching
argument for the mean-field model. A principal difference between that
argument and those of \cite{craw-i,ioffe_geom} and the current work is that it uses 
the classical switching lemma developed in \cite{aiz82}, applied to a discretized version
of the mean-field system.
\end{remark}

\section{Classical and quantum Ising models}\label{sec-backg}

Let $L=(V,E)$ be a finite, connected graph, which (for simplicity only) we assume
to have neither loops nor multiple edges. An edge of $L$
with endpoints $u$, $v$ is denoted by $uv$. We write $u \sim v$ if $uv \in E$.

\subsection{Quantum Ising model with transverse field}\label{ssec-qIm}
As basis for each copy of $\CC^2$ in the Hilbert space 
$\cH = \bigotimes_{x\in V} \CC^2 $, we take the vectors 
$|+_v\er=\big(\begin{smallmatrix} 1 \\ 0\end{smallmatrix}\big)$ and 
$|-_v\er=\big(\begin{smallmatrix} 0 \\ 1\end{smallmatrix}\big)$.
Let $D$ be the set of $2^{|V|}$ basis vectors of $\cH$ of the form
$\ket{\s}=\bigotimes_{v\in V}\ket{\pm_v}$. There is a natural one--one
correspondence between $D$ and the space $\Si=\{-1,+1\}^V$, and we shall speak of
$\cH$ as being generated by $\Si$. The \emph{trace} of the Hermitian matrix $A$ is
defined as
$$
\tr(A) = \sum_{\s\in\Si} \bra{\s}A\ket{\s}.
$$
Here, $\bra{\psi}$ is the adjoint, or complex transpose, of 
the vector $\ket{\psi}$.

The Hamiltonian of the quantum
Ising model with transverse field is given in~\eqref{qi_ham_eq}.  
Let $\b>0$ be a fixed real number (known as the `inverse temperature'), and 
define the positive temperature states
\begin{equation}\label{qi_states_eq}
\oper_{L,\b}(Q)=\frac{1}{Z_L(\b)}\tr(e^{-\b H}Q),
\end{equation}
where $Z_L(\b)=\tr(e^{-\b H})$ and $Q$ is a suitable matrix.  The
\emph{ground state} is defined as the limit $\oper_L$ of $\oper_{L,\b}$ as
$\b\rightarrow\infty$.  If $(L_n: n \ge 1)$ is an increasing sequence of graphs
tending to an infinite
graph $L$, then we may also make use of the \emph{infinite-volume} limits
$$
\oper_{L,\b}=\lim_{n\rightarrow\infty}\oper_{L_n,\b},\qq
\oper_L=\lim_{n\rightarrow\infty}\oper_{L_n}.
$$
The existence of such limits is
discussed in~\cite{akn}.

\subsection{Space--time Ising model}\label{ssec-stI}
A number of authors have developed and utilized the following `path integral
representation' of the 
quantum Ising model, see for
example~\cite{akn,aizenman_nacht,campanino_klein_perez,chayes_ioffe_curie-weiss,GOS,nachtergaele93} 
and the recent surveys to be found in \cite{G-pgs,ioffe_geom}.
Let $\SS=\SS_\b$ be the circle of circumference $\b$, which we 
think of as the interval $[0,\b]$ with its two 
endpoints identified.    Let $\l,\d,\g$ be non-negative constants, and let
$\mu_\l$, $\mu_\d$, $\mu_\g$ 
be the probability measures associated with independent Poisson processes on
$E\times\SS$, $V\times\SS$, and  $V\times\SS$ with respective intensities $\l,\d,\g$.
Elements sampled from these measures
will typically be denoted by $B$, $\oD$, $G$, and their members
will be called \emph{bridges}, 
\emph{deaths} and \emph{ghost-bonds} respectively. 

\begin{remark}\label{rem-as}
For simplicity of notation in this article, we shall frequently overlook
events with zero probability. 
\end{remark}

Thus, for example, we shall assume without more ado that the $\SS$ coordinates 
of the points of $B \cup \oD \cup G$
are distinct. Furthermore,
we shall take as sample space for $B$ (\resp, $\oD$, $G$) 
the set $\cB$ (\resp, $\cF$) of
finite subsets of $E\times\SS$ (\resp, $V\times\SS$).  

For $\oD \in \cF$,    
write $V(\oD)$ for the collection of maximal intervals of
$(V\times\SS)\setminus \oD$, and let $\S(\oD)=\{-1,+1\}^{V(\oD)}$.  
Each $\s\in\S(\oD)$ should be viewed as a spin-configuration on $(V\times\SS)\sm \oD$ using
local spins $\pm 1$: for $x =(v,t) \in (V\times\SS)\sm \oD$, write $\s_x=\s_{(v,t)}$ for the 
local state of $x$ under $\s$, that is, the $\s$-value of the interval
in $V(D)$ containing $x$.  Note that $\s_x$ is undefined for $x \notin \oD$, but, since $D$
is almost surely finite, this is no bar to the following definition.

The \emph{space--time Ising measure} on the domain 
$$
\L:= L\times \SS=(V\times\SS,E\times\SS)
$$ 
is defined to be
the probability measure on the space 
$$
\Si=\bigcup_{D \in \cF} \Si(D),
$$
with partition function 
\begin{equation}
Z'=\int_\cF d\mu_\d(\oD)\sum_{\s\in\S(\oD)}\exp\left\{\l\int_{E\times\SS}\s_e\, de+
\g\int_{V\times\SS}\s_x\, dx\right\}
\label{ihp6}
\end{equation}
where $\s_e=\s_{(u,t)}\s_{(v,t)}$ if $e=(uv,t)$.
The two integrals in \eqref{ihp6} are to be interpreted, respectively, as
$$
\sum_{e=uv\in E} \int_\SS \s_{(u,t)}\s_{(v,t)}\,dt,\qq \sum_{v\in V}\int_\SS \s_{(v,t)}\,dt.
$$
Note that the use of the \emph{circle} $\SS$ amounts to a periodic boundary
condition in the $\b$ direction. We shall generally suppress reference to $\b$ in the
following.

Here is a word of motivation for \eqref{ihp6}; see also \cite{bjo_phd,GOS}. 
Let $D\in\cF$, and think of $V(D)$ as the set of vertices of a graph with edges
given as follows.  We
augment $V(D)$ with an auxiliary vertex, called the \emph{ghost-vertex}
and denoted $\Gh$, to which we assign spin $\s_\Gh=1$.
An edge is placed between $\Gh$ and each $\bar v\in V(D)$. For $\bar u, \bar v\in V(D)$,
with $\bar u=u\times I_1$ and $\bar v=v\times I_2$ say, we place an edge 
between $\bar u$ and $\bar v$ if and only if: (i) $uv$ is an edge of $L$, and
(ii) $I_1\cap I_2\neq\es$.  Under the measure
with partition function \eqref{ihp6}, and conditional on $D$, a spin-configuration
$\s\in \Si(D)$ on this graph
receives an Ising weight
\begin{equation}
\exp\left\{\sum_{\bar u\bar v}J_{\bar u\bar v}\s_{\bar u}\s_{\bar v}  
+ \sum_{\bar v} h_{\bar v}\s_{\bar v}\right\},
\end{equation}
where $\s_{\bar v}$ denotes the common value of $\s$ along $\bar v$, and with
$J_{\bar u\bar v}=\l|I_1\cap I_2|$ and $h_{\bar v}=\g |\bar v|$.
Here, $|J|$ denotes the Lebesgue measure of the interval $J$.
This observation will be pursued further in Section~\ref{ssec-rpr}.

We will use angle brackets $\el\cdot\er$ for the expectation operator under
the measure given by \eqref{ihp6}. Thus, for example,
\begin{equation}\label{st_Ising_eq}
\el\s_A\er=\frac{1}{Z'}
\int d\mu_\d(\oD)\sum_{\s\in\S(\oD)}\s_A\exp\left\{\l\int_{E\times\SS}\s_e\, de+
\g\int_{V\times\SS}\s_x\, dx\right\},
\end{equation}
where $A\subseteq V\times\SS$ is a finite set, and 
\begin{equation}
\s_A:=\prod_{y\in A}\s_y.
\end{equation}

Let $0$ be a given point of $V\times \SS$.
We will be particularly concerned with the
\emph{magnetization}  and \emph{susceptibility}
of the space--time Ising model on $\L=L\times \SS$, given respectively by
\begin{align}
M=M_\L(\l,\d,\g) &:=\el\s_0\er,\label{def-mag}\\ 
\chi=\chi_\L(\l,\d,\g) &:=\frac{\partial M}{\partial \g}
=\int_\L\el\s_0;\s_x\er\,dx,
\label{def-susc}
\end{align}
where the \emph{truncated} two-point function $\el\s_0;\s_x\er$ is given by
\begin{equation}
\el\s_A;\s_B\er := \el \s_A\s_B\er - \el\s_A\er\el\s_B\er.
\label{o20}
\end{equation}

We will derive a number of differential inequalities for $M$
and $\chi$, of which the following is the principal one. In writing 
$L=[-n,n]^d$, we mean that $L$ is the box $[-n,n]^d$ of $\ZZ^d$ with
`periodic boundary conditions', which is to say that
two vertices $u$, $v$ are joined by an edge whenever there exists
$i\in\{1,2,\dots,d\}$ 
such that: $u$ and $v$ differ by exactly $2n$ in the $i$th 
coordinate, and the other coordinates are equal. 
(Our results are in fact valid in greater generality,
see the statement before Assumption \ref{periodic_assump}.) Subject to this boundary condition, $M$
and $\chi$ do not depend on the choice of origin $0$.

\begin{theorem}\label{main_pdi_thm}
Let $d\ge 1$ and let
$L = [-n,n]^d$.  Then
\begin{equation}\label{ihp18}
M\leq \g\chi+M^3+2\l M^2\frac{\partial M}{\partial \l}
-2\d M^2\frac{\partial M}{\partial \d}.
\end{equation}
\end{theorem}

A similar inequality was derived in~\cite{abf} for the classical Ising
model, and our method of proof is closely related to that used there. Other such inequalities
have been proved for percolation in \cite{ab} 
(see also \cite{grimmett_perc}), and for the contact model
in \cite{aizenman_jung,bezuidenhout_grimmett}.
As observed in \cite{ab,abf}, the powers of $M$ on the right side 
of \eqref{ihp18} determine the bounds of Theorems \ref{ed_thm}(ii) and \ref{mf_thm}  on
the critical exponents.
The cornerstone of our proof
is a random-parity representation of the space--time Ising model.

In the ground-state limit as $\b,n\to\oo$ and $\g > 0$, 
the two quantities $M$, $\chi$ have well-defined limits denoted $M_\oo$ and $\chi_\oo$.
By a re-scaling argument, $M_\oo$ depends on the parameters through the ratios $\l/\d$, $\g/\d$.  
Thus we may take as `order parameter' the function
$$
M(\rho,\g):= M_\oo(\rho,1,\g).
$$
More generally, let $M^\b(\rho,\g) = M_\oo^\b(\rho,1,\g)$ where
$M_\oo^\b = \lim_{n\to\oo} M^\b$,
and define 
the critical value $\bc^\b$ by \eqref{o1}.

The analysis of the differential inequalities, following \cite{ab,abf}, reveals a
number of facts about the behaviour of the model.  In
particular, we will show the exponential
decay of the correlations $\el\s_0\s_x\er_+^\b$
when $\rho<\bc^\b$ and $\g=0$, 
as asserted in Theorem \ref{ed_thm}, and in addition
certain bounds on two critical exponents of the model.  See
Section~\ref{cons_sec} for further details. 

We shall on occasion write $\mu(f)$ for the expectation of a random variable $f$ under 
the probability measure $\mu$.
The indicator function of an event $H$ is written either $1_H$ or $1\{H\}$. The complement
of $H$ is written $H^\tc$.

\subsection{Classical/quantum relationship}\label{ssec-cq}
The space--time Ising model 
is closely related to the quantum Ising model, one
manifestation of this being the following.  As indicated at the start of this section,
a classical spin configuration
$\s\in\Si=\{-1,+1\}^V$ may be identified with the basis vector 
$|\s\er=\bigotimes_{v\in V}|\s_v\er$ of $\cH$.  The state $\oper_{L,\b}$
of \eqref{qi_states_eq} gives rise thereby to a probability measure $\mu$ on
$\Si$ by 
\begin{equation}
\mu(\s)=\frac{\el\s|e^{-\b H}|\s\er}{\tr(e^{-\b H})},\qq \s\in\Si.
\end{equation}
When $\g=0$, it turns out that $\mu$ is the
law of the vector $(\s_{(v,0)}: v\in V)$ under the space--time Ising measure
of \eqref{ihp6} (see \cite{akn} and the references therein).
It therefore makes sense to study the phase diagram of the quantum Ising model
via its representation in the space--time Ising model.
Note, however, that in our analysis it is crucial to work with $\g>0$,
and to take the limit $\g \downarrow 0$ later.  
The role played in the classical model by the external
field will in our analysis be played by the `ghost-field' $\g$ rather than
the `physical' transverse field $\d$. 

We draw from \cite{akn,aizenman_nacht} in the following summary of the relationship between the 
phase transitions of the quantum and space--time Ising models. Let $u,v\in V$, and 
$$
\tau^\b_{L}(u,v) := \tr\bigl(\oper_{L,\b}(Q_{u,v})\bigr),\qq
Q_{u,v} = \s^{(3)}_u\s^{(3)}_v.
$$
It is the case that
\begin{equation}
\tau^\b_{L}(u,v) = \el \s_A \er^\b_L
\label{o3}
\end{equation}
where $A=\{(u,0),(v,0)\}$, and the role of $\b$ is emphasized
in the superscript. Let $\tau_L^\oo$ denote the limit
of $\tau^\b_{L}$ as $\b\to\oo$. For $\b\in(0,\oo]$, let $\tau^\b$
be the limit of $\tau^\b_L$ as $L\uparrow \ZZ^d$.
(The existence of this limit may depend on the choice of boundary condition on $L$,
and we return to this at the end of Section \ref{cons_sec}.)
By Theorem \ref{ed_thm},
\begin{equation}
\tau^\b(u,v) \le c'e^{-c|u-v|},
\label{o4}
\end{equation}
where $c'$, $c$ depend on $\rho$, and $c>0$ for $\rho<\bc^\b$ and $\b\in(0,\oo]$.
Here, $|u-v|$ denotes the $L^1$ distance from $u$ to $v$. The situation when $\rho=\bc^\b$ is
more obscure, but one has that
\begin{equation}
\limsup_{|v|\to\oo}\tau^\b(u,v) \le M^\b_+(\rho),
\label{o4a}
\end{equation}
so that $\tau^\b(u,v) \to 0$ as $|v|\to\oo$, whenever $M^\b_+(\rho)=0$. It 
is proved at Theorem \ref{crit_val_cor} that $\bc^\oo=2$ and $M^\oo_+(2)=0$ when $d=1$.

By the {\fkg} inequality, and the uniqueness
of infinite clusters in the continuum \rc\ model (see \cite{akn,grimmett_stp},
for example), 
\begin{equation}
\tau^\b(u,v) \ge M^\b_+(\rho-)^2 > 0,
\label{o5}
\end{equation}
when $\rho > \bc^\b$ and $\b\in(0,\oo]$, where $f(x-):= \lim_{y\uparrow x}f(y)$.
The proof is discussed at the end of Section \ref{cons_sec}.

The quantum
mean-field, or Curie--Weiss, model has been studied using 
large-deviation techniques in \cite{chayes_ioffe_curie-weiss}, see also
\cite{grimmett_stp}.  A 
random-current representation of the quantum Ising model 
may be found in~\cite{ioffe_geom}, and, as explained in
Remark \ref{crawford} and \cite{craw-i}, this is intimately
related to that discussed and exploited in the next section. 

\section{The random-parity representation}\label{rcr_sec}

The Ising model on a discrete graph $L$ is a `site model', in the sense that
configurations comprise spins assigned to the vertices (or `sites') of
$L$.  The classical random-current representation maps this into a bond-model,
in which the 
sites no longer carry random values, but instead the \emph{edges} $e$ (or `bonds')  of
the graph are replaced by a random number $N_e$ of parallel edges. The
bond $e$ is called \emph{even} (\resp, \emph{odd}) if $N_e$ is
even (\resp, odd). The odd bonds may be arranged into paths and
cycles.  One cannot proceed
in the same way in the above space--time
Ising model.  

There are two 
possible alternative
approaches.  The first uses the fact that,
conditional on the set $\oD$ of deaths, $\L$ may be viewed as a discrete structure 
with finitely many components, to which the 
random-current representation of \cite{aiz82} may be applied;  this is explained in
detail around \eqref{rcr_step1_eq} below. Another approach is to forget about `bonds', and
instead to concentrate on the parity configuration associated with a current-configuration, as
follows.  The relationship with the random-current representation of \cite{ioffe_geom}
is discussed in Remark \ref{crawford}.  

The circle $\SS$ may be viewed as a continuous limit of a ring of 
equally spaced points.  If we apply the random-current
representation to the discretized system, but only record whether a bond is
even or odd, the representation has a well-defined limit as a
partition of $\SS$ into even and odd sub-intervals.  In the limiting
picture, even and odd intervals carry different weights, and it is the
properties of these weights that render the representation useful.  This is
the essence of the main result in this section, 
Theorem~\ref{rcr_thm}. We will prove this result without recourse to
discretization.  

\subsection{Colourings}\label{ssec-col}
We first generalize the set-up of Section \ref{sec-backg}.
For $v\in V$, let $K_v\subseteq\SS$ be a finite union of (maximal)
disjoint
intervals, 
say $K_v=\bigcup_{i=1}^{m(v)} I^v_i$.
No assumption is made at this stage on whether the $I^v_i$ are
open, closed, or half-open.
For $e=uv\in E$, let $K_e=K_u\cap K_v$.  With the $K_v$ given, we define
\begin{gather}
K:=\bigcup_{v\in V} v\times K_v,\qq
F:=\bigcup_{e\in E} e\times K_e,
\label{o8}\\
\L:=(K,F),
\label{def-newL}
\end{gather}
where these sets are considered as unions of real intervals.
We shall soon introduce an auxiliary `ghost-vertex', denoted
$\Gh$, and shall write 
\begin{equation}
K^\Gh := K \cup \{\Gh\}.
\label{ihp0}
\end{equation}
In Section \ref{sec-backg}, we treated only the case when each 
$K_v$ comprises the single interval $\SS:=[0,\b]$.
We continue to use the notation $\cB$ (\resp, $\cF$) for the set of finite subsets of
$F$ (\resp, $K$). The closure of a Borel subset $J$ of $\ZZ\times \RR$ is written $\ol J$.

Much of the following analysis is valid with the constants  $\l$, $\d$, $\g$ 
replaced by (possibly non-constant) functions.
Specifically, let
$\l:E\times\SS\rightarrow\RRp$, $\d:V\times\SS\rightarrow\RRp$, and 
$\g:V\times\SS\rightarrow\RRp$ be
bounded, measurable functions, where $\RRp=[0,\oo)$.  We retain the notation $\l$, $\d$, $\g$
for the restrictions of these functions to $\L$, given in \eqref{def-newL},
and let $\mu_\l$, $\mu_\d$, $\mu_\g$
be the probability measures associated with independent Poisson processes with
respective intensities $\l$, $\d$, $\g$ on the respective subsets of $\L$.
For $\oD\in\cF$, the set $(v\times K_v)\sm\oD$ is
a union of maximal death-free intervals $v\times J_v^k$,
where $k=1,2,\dotsc,n$ and $n=n(v,\oD)$ is the number of such intervals.  With
$V(\oD)$ the collection of all such intervals, and 
$\S(\oD)=\{-1,+1\}^{V(\oD)}$ as before, we may define the space--time Ising measure
on the $\L$ of \eqref{def-newL} as that with partition function
\begin{equation}
Z'_K=\int_\cF d\mu_\d(\oD)\sum_{\s\in\S(\oD)}\exp\left\{\int\l(e)\s_e\, de+
\int\g(x)\s_x\, dx\right\}.
\label{o12}
\end{equation}
As in \eqref{st_Ising_eq}, we write $\el \s_A\er_K$, abbreviated to $\el\s_A\er$ when the
context is obvious,
for the mean of $\s_A$ under this measure.

It is essential for our method that we work on general
domains of the form given in \eqref{def-newL}. The reason for this
is that, in the geometrical analysis of currents, we shall at times
remove from $K$  a random subset called the `backbone', and the ensuing
domain has the form of \eqref{def-newL}.
This generalization also allows us to work with a `free'
rather than a `vertically periodic' boundary condition.  
That is, by setting $K_v =[0,\b)$ for all $v\in V$, rather than
$K_v = [0,\b]$, we effectively remove the restriction that the
`top' and `bottom' of each $v\times\SS$ have the same spin.  

Whenever we wish to
emphasize the roles of particular $K$, $\l$, $\d$, $\g$, we include them as subscripts.  For
example, we may 
write $\el\s_A\er_K$ or $\el\s_A\er_{K,\g}$ or $Z'_\g$, and so on.

We now define two additional random processes associated with the
space--time Ising measure on $\L$.  The first is a random colouring of $K$,
and the second is a random (finite) weighted graph.  These two objects
will be the main components of the random-parity representation.

Let $\ol K$ be the closure of $K$. A set of \emph{sources} is a finite set
$A \subseteq \ol K$ such that: each $a \in A$ is the endpoint of at most one maximal sub-interval
$I_i^v$ of $K$.
(This last condition is for simplicity later.)
Let $B\in\cB$ and $G\in\cF$.  
Let $S=A \cup G \cup V(B)$, where $V(B)$ is the set of
endpoints of bridges of $B$,
and call members of $S$ \emph{switching points}. As in Remark \ref{rem-as},
we shall assume that $A$, $G$ and $V(B)$ are disjoint.

We
shall define a colouring $\psi^A=\psi^A(B,G)$ of $K\sm S$ using the two colours (or
labels) `even' and `odd'.  This colouring is constrained to be `valid',
where
a valid colouring is defined to be a mapping $\psi:K\sm S \to\{\even, \odd\}$ such that: 
\begin{romlist}
\item the label is
constant between 
two neighbouring switching points, that is, $\psi$ is constant on any sub-interval
of $K$ containing no members of $S$,
\item  the label always switches at each switching point, which is to say that,
for $(u,t) \in S$, $\psi(u,t-) \ne \psi(u,t+)$, whenever these two values are defined,

\item for any pair $v$, $k$ such that $I^v_k\neq\SS$, in the limit as we move
along $v\times I^v_k$ towards an endpoint $a$ of $v\times I^v_k$, 
the colour converges to `even' if $a \notin A$, and to `odd' if $a\in A$.
\end{romlist} 

If there exists $v\in V$ and
$1\leq k\leq m(v)$ such that $v\times \ol{I^v_k}$
contains an
\emph{odd} number of switching points, then conditions (i)--(iii) cannot be
satisfied;  in this case we set the colouring $\psi^A$ to
a default value denoted $\#$.  

Suppose that (i)--(iii) \emph{can} be satisfied, and let 
\begin{equation*}
W=W(K):=\{v\in V: K_v=\SS\}. 
\end{equation*}
If $W = \es$ (in which case we speak of a `free'
boundary condition), then there exists a unique valid colouring, 
denoted $\psi^A$.  If $r=|W|\ge 1$, there are exactly $2^r$ valid
colourings, one for each of the two possible colours assignable to the sites
$(w,0)$, $w \in W$; in this case we let $\psi^A$ be chosen uniformly at random
from this set, independently of all other choices. (If $(w,0)\in S$, we work instead
with the colour of $(w,\eps)$ in the limit as $\eps\downarrow 0$.)

Let $M_{B,G}$ be the probability measure (or expectation when
appropriate) governing the randomization in the definition of $\psi^A$: $M_{B,G}$ is
the uniform (product) measure on the set of valid colourings, and it is a point
mass if and only if $W=\es$.  See Figure~\ref{colouring_fig}.  

Fix the set $A$ of sources.
For (almost every) pair $B$, $G$, one may construct as above
a (possibly random) colouring $\psi^A$. Conversely, it is easily seen that
the pair $B$, $G$ may (almost surely) be reconstructed from 
knowledge of the colouring $\psi^A$.
For given $A$, we may thus speak of a configuration as being either a pair $B$, $G$, or 
a colouring $\psi^A$. 
While $\psi^A(B,G)$ is a colouring of $K \sm S$ only, we shall
sometimes refer to it as a colouring of $K$.

\begin{figure}[tbp]
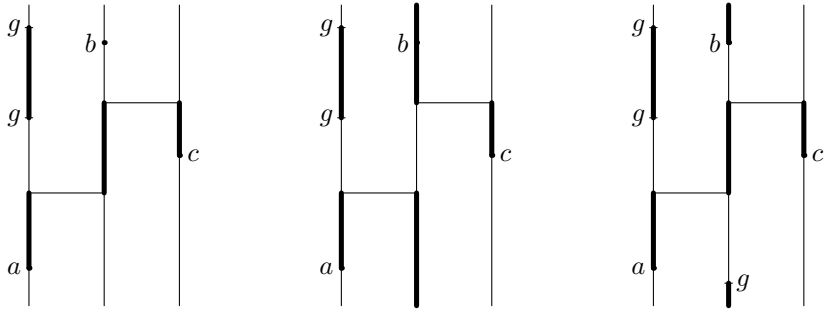

\includegraphics{sharptransition.1}
\hspace{1.3cm} 
\includegraphics{sharptransition.2} 
\hspace{1.3cm} 
\includegraphics{sharptransition.3} 
\caption{Three examples of colourings for given $B \in\cB$, $G \in \cF$.
Points in $G$ are written $g$. Thick line
segments are `odd' and thin segments `even'.  In this illustration we have
taken $K_v=\SS$ for all $v$.  \emph{Left and middle}: two of the eight
possible colourings when the sources are $a$, $c$.  \emph{Right}: one of the possible
colourings when the sources are $a$, $b$, $c$.}
\label{colouring_fig}
\end{figure}

The next step is to assign weights $\pd\psi$ to colourings $\psi$. The
`failed' colouring $\#$ is assigned weight $\pd\# =0$.  
For every valid colouring $\psi$, let $\ev(\psi)$ (\resp, $\odd(\psi)$)
denote the subset of $K$ that is labelled even (\resp, odd), 
and let
\begin{equation}
\partial\psi :=\exp\bigl\{2\d(\ev(\psi))\bigr\},
\label{def-wt}
\end{equation}
where
$$
\d(U):=\int_{U}\d(x)\,dx, \qq U \subseteq V\times \SS.
$$
Up to a multiplicative constant depending on $\d(K)$ only,
$\pd\psi$ equals the square of the probability that the odd part of $\psi$ is
death-free.

\subsection{Random-parity representation}\label{ssec-rpr}
The expectation $E(\pd\psi^A)$ is taken over the sets $B$, $G$, and over
the randomization that takes place when $W \ne \es$, 
that is, $E$ denotes expectation with respect to the
measure $d\mu_\l(B) d\mu_\g(G) dM_{B,G}$. The notation has been chosen to harmonize with that
used in \cite{abf} in the discrete case:
the expectation $E(\partial\psi^A)$ will play the role of the probability
$P(\partial\underline n=A)$ of \cite{abf}.
The main result of this section now follows.

\begin{theorem}[Random-parity representation]\label{rcr_thm}
For any finite set $A \subseteq \ol K$ of sources,
\begin{equation}
\el\s_A\er=\frac{E(\partial\psi^A)}{E(\partial\psi^\es)}.
\label{ihp8}
\end{equation}
\end{theorem}

We introduce a second random object in advance of proving this.
Let $D \in \cF$, the set of finite subsets of $K$,
and recall that $K\sm \oD$ is a disjoint union of intervals
of the form $v\times J_v^k$.
For each $e=uv\in E$, and each $1\leq k\leq n(u)$ and $1\leq l\leq n(v)$, let 
\begin{equation}
J^e_{k,l}:=J^u_k\cap J^v_l,
\end{equation}
and 
\begin{equation}
E(\oD)=\bigl\{e\times J^e_{k,l}:e\in E,\ 1\leq k\leq n(u),\ 1\leq l\leq n(v),
\,J^e_{k,l}\neq\es\bigr\}.
\end{equation}
Up to a finite set of points, $E(\oD)$  forms a partition of the set
$F$ induced by the `deaths' in $\oD$.  

\begin{figure}[tbp]
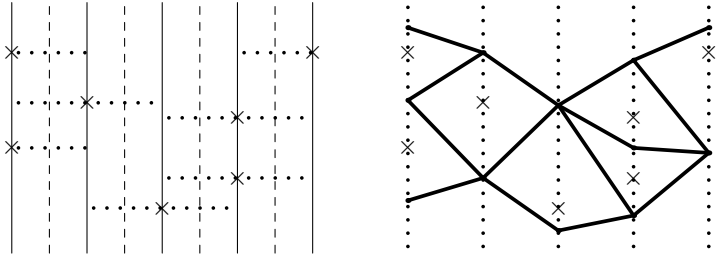

\includegraphics{sharptransition.4}
\qquad
\includegraphics{sharptransition.5} 
\caption{\emph{Left}: 
  The partition $E(\oD)$.  We have: $K_v=\SS$ for $v \in V$, the lines $v\times K_v$
  are drawn as solid, the lines $e\times K_e$ as dashed, and elements of $\oD$ are
  marked as crosses.  The endpoints of the $e\times J^e_{k,l}$ are the points
  where the dotted lines meet the dashed lines.
  \emph{Right}:  
  The graph $G(\oD)$.  In this illustration, the dotted lines are the
  $v\times K_v$, and the solid lines are the edges of $G(\oD)$.}
\label{bridge_part_fig}
\end{figure}

The pair
\begin{equation}
G(\oD):=(V(\oD),E(\oD))
\end{equation}
may be viewed as a graph, illustrated in Figure~\ref{bridge_part_fig}.
We will use the symbols $\bar v$ and $\bar e$ for
typical elements of $V(\oD)$ and $E(\oD)$, respectively.
There are natural weights on the edges and vertices of $G(\oD)$: for
$\bar e=e\times J^e_{k,l}\in E(\oD)$ and
$\bar v=v\times J^v_k\in V(\oD)$, let
\begin{equation}
J_{\bar e}:= \int_{J^e_{k,l}}\l(e,t)\,dt,
\qquad 
h_{\bar v}:= \int_{J^v_k}\g(v,t)\,dt.
\label{o10}
\end{equation}
Thus the weight of a vertex or edge is its measure, calculated according to
$\l$ or $\g$, respectively. By \eqref{o10},
\begin{equation}
\sum_{\bar e\in E(\oD)} J_{\bar e} + \sum_{\bar v\in V(\oD)}h_{\bar v}
= \int_{F}\l(e)\,de + \int_{K}\g(x)\,dx,\qq D\in \cF.
\label{o11}
\end{equation}

\begin{proof}[Proof of Theorem~\ref{rcr_thm}]
With $\L=(K,F)$ as in \eqref{def-newL}, we consider the partition function
$Z'=Z'_K$ given in \eqref{o12}.
For each $\bar v\in V(\oD)$,
$\bar e\in E(\oD)$, the spins $\s_v$ and $\s_e$ are constant for $x\in\bar v$
and $e\in\bar e$, respectively.  Denoting their common values by 
$\s_{\bar v}$ and $\s_{\bar e}$ respectively, the summation in \eqref{o12} equals
\begin{multline}
\sum_{\s\in\S(\oD)}\exp\left\{
\sum_{\bar e\in E(\oD)}\s_{\bar e}\int_{\bar e}\l(e)\, de+
\sum_{\bar v\in V(\oD)}\s_{\bar v}\int_{\bar v}\g(x)\, dx\right\}\\
=\sum_{\s\in\S(\oD)}\exp\left\{
\sum_{\bar e\in E(\oD)}J_{\bar e}\s_{\bar e}+
\sum_{\bar v\in V(\oD)}h_{\bar v}\s_{\bar v}\right\}.
\label{ihp2}
\end{multline}
The right side of \eqref{ihp2} is the partition function of the discrete Ising model on
the graph $G(\oD)$, with pair couplings $J_{\bar e}$ and external fields
$h_{\bar v}$.  We shall apply the random-current expansion of \cite{abf} to this
model.  

For convenience of exposition, we introduce the extended graph
\begin{align}
\wtilde G(\oD)&=(\wtilde V(\oD),\wtilde E(\oD))\label{ihp3}\\
&:=
\bigl(V(\oD)\cup\{\Gh\},E(\oD)\cup\{\bar v \Gh: \bar v \in V(\oD)\}\bigr)
\nonumber
\end{align}
where $\Gh$ is the
ghost-site of \eqref{ihp0}.  We call members of $E(\oD)$ \emph{lattice-bonds},
and those of $\wtilde E(\oD)\sm E(\oD)$ \emph{ghost-bonds}.  
Let $\Psi(\oD)$ be the random multigraph with vertex set $\wtilde V(\oD)$ and
with each edge of $\wtilde E(\oD)$ replaced by a random number of parallel edges,
these numbers being independent and having the Poisson distribution, with
parameter $J_{\bar e}$ for lattice-bonds $\bar e$, and parameter $h_{\bar v}$
for ghost-bonds $\bar v\Gh$.  

Let $\{\partial\Psi(\oD)=A\}$ denote the event
that, for each $\bar v\in V(\oD)$, the total degree of $\bar v$ in $\Psi(\oD)$
\emph{plus} the number of elements of $A$ inside the closure
of $\bar v$ (when regarded as an
interval) is even.  
(There is $\mu_\d$-probability $0$ that $A\cap D \ne\es$,
and thus we may overlook this possibility.) 
Applying the discrete random-current expansion, and
in particular~\cite[eqn (9.24)]{grimmett_RCM}, we obtain by \eqref{o11} that
\begin{equation}
\sum_{\s\in\S(\oD)}\exp\left\{
\sum_{\bar e\in E(\oD)}J_{\bar e}\s_{\bar e}+
\sum_{\bar v\in V(\oD)}h_{\bar v}\s_{\bar v}\right\}=
c 2^{|V(\oD)|}P_D(\partial\Psi(\oD)=\es),
\end{equation}
where $P_D$ is the law of the edge-counts, and 
\begin{equation}
c=\exp\left\{\int_F \l(e)\,de + \int_K\g(x)\,dx\right\}.
\label{o21}
\end{equation}

By the same argument applied to the numerator in \eqref{st_Ising_eq}
(adapted to the measure on $\L$, see the remark after \eqref{o12}),
\begin{equation}
\label{rcr_step1_eq}
\el\s_A\er=
\frac{E(2^{|V(\oD)|}1\{\partial\Psi(\oD)=A\})}
{E(2^{|V(\oD)|}1\{\partial\Psi(\oD)=\es\})},
\end{equation}
where the expectation is with respect to $\mu_\d \times P_D$. 
The claim
of the theorem will follow by an appropriate manipulation of \eqref{rcr_step1_eq}. 

Here is another way to sample $\Psi(\oD)$ which allows us to couple it with
the random colouring $\psi^A$.  Let $B\in\cB$ and $D,G\in\cF$.
The number of points of $G$ lying
in the interval $\bar v=v\times J^v_k$ has the Poisson distribution with
parameter $h_{\bar v}$, and similarly the number of elements of $B$ lying
in $\bar e=e\times J^e_{k,l}\in E(\oD)$ has the Poisson distribution with
parameter $J_{\bar e}$.  Thus, for given $\oD$, the multigraph $\Psi(B,G,\oD)$,
obtained by replacing an edge of $\wtilde E(\oD)$ by parallel edges
equal in number to the corresponding number
of points from $B$ or $G$, respectively, has the same law as $\Psi(\oD)$.  Using the
\emph{same} sets $B$, $G$ we may form the random colouring $\psi^A$.  

The
numerator of~\eqref{rcr_step1_eq} satisfies
\begin{align}
&E(2^{|V(\oD)|}1\{\partial\Psi(\oD)=A\})\label{ihp5}\\
&\hskip1cm=\iint d\mu_\l(B)\, d\mu_\g(G)\,\int d\mu_\d(\oD)\, 
2^{|V(\oD)|}1\{\partial\Psi(B,G,\oD)=A\}\nonumber\\
&\hskip1cm= \mu_\d(2^{|V(D)|}) 
\iint d\mu_\l(B)\, d\mu_\g(G)\,\wtilde\mu(\partial\Psi(B,G,\oD)=A),
\nonumber
\end{align}
where $\wtilde\mu$ is the probability measure on $\cF$ satisfying
\begin{equation}
\frac{d\wtilde\mu}{d\mu_\d}(D) \propto 2^{|V(\oD)|}.
\label{ihp10}
\end{equation}
Therefore, by \eqref{rcr_step1_eq},
\begin{equation}
\el\s_A\er=
\frac{\wtilde P(\partial\Psi(B,G,\oD)=A)}
{\wtilde P(\partial\Psi(B,G,\oD)=\es)},
\label{ihp9}
\end{equation}
where $\wtilde P$ denotes the probability under $\mu_\l\times\mu_\g\times\wtilde \mu$.
We claim that
\begin{equation}
\wtilde \mu(\pd\Psi(B,G,\oD)=A) = s M_{B,G}(\pd\psi^A(B,G)),
\label{ihp13}
\end{equation}
for all $B$, $G$, where $s$ is a constant,
and the expectation $M_{B,G}$ is over the uniform measure on the set of valid colourings. 
Claim \eqref{ihp8} follows from this,
and the remainder of the proof is to show \eqref{ihp13}.
The constants $s$, $s_j$ are permitted in the following to depend only on
$\L$ and $\d$.  

Here is a special case. For $B\in\cB$, $G\in\cF$, 
\begin{equation}
\wtilde\mu(\partial\Psi(B,G,\oD)=A)=0
\end{equation}
if and only if some interval $\ol{I^v_k}$
 contains an odd number of switching
points, if and only if $\psi^A(B,G) =\#$ and $\partial\psi^A(B,G)=0$.  
Thus \eqref{ihp13} holds in this case.

Another special case arises 
when $K_v=[0,\b)$ for all $v\in V$, that is, the `free boundary' case.  
Assume that each $\ol{K_v}$ contains an even number of switching points.
As remarked earlier,
there is a unique valid colouring $\psi^A=\psi^A(B,G)$.
Moreover, $|V(\oD)|=|\oD|+|V|$, whence from standard properties of Poisson
processes, $\wtilde\mu=\mu_{2\d}$.  It may be seen after some thought 
(possibly with the aid of a diagram) that, for given $B$, $G$,
the events $\{\partial\Psi(B,G,\oD)=A\}$ and $\{\oD\cap\odd(\psi^A)=\es\}$
differ by an event of $\mu_{2\d}$-probability $0$. Therefore,
\begin{align}
\wtilde\mu(\partial\Psi(B,G,\oD)=A)&=
\mu_{2\d}( \oD\cap\odd(\psi^A)=\es )\label{ihp7}\\
&=\exp\{-2\d(\odd(\psi^A))\}\nonumber\\
&= s_1\exp\{2\d(\ev(\psi^A))\}=s_1 \partial\psi^A,
\nonumber
\end{align}
with $s_1=e^{-2\d(K)}$. In this special case, \eqref{ihp13} holds.

For the general case, we first note some properties of
$\wtilde\mu$.  By the above, we may assume that
$B$, $G$ are such that $\wtilde\mu(\partial\Psi(B,G,\oD)=A)>0$,
which is to say that each $\ol{I_k^v}$
contains an even number of switching points. 
Let $W = \{v\in V: K_v=\SS\}$ and,  for
$v\in V$, let $\oD_v=D\cap(v\times K_v)$ and $\od(v)=|\oD_v|$.  By \eqref{ihp10}, 
\begin{align*}
\frac{d\wtilde\mu}{d\mu_\d}(D) \propto 2^{|V(\oD)|}&=
\prod_{w\in W}2^{1\vee \od(w)}\prod_{v\in V\setminus W}2^{m(v)+\od(v)}\\
&\propto 2^{|\oD|} \prod_{w\in W} 2^{1\{\od(w)=0\}},
\end{align*}
where $a\vee b = \max\{a,b\}$, and we recall the number $m(v)$ of intervals $I^v_k$ that constitute
$K_v$.  Therefore,
\begin{equation}
\frac{d\wtilde\mu}{d\mu_{2\d}}(D) \propto \prod_{w\in W} 2^{1\{\od(w)=0\}}.
\end{equation}
Three facts follow.
\begin{letlist}
\item The sets $\oD_v$, $v\in V$ are independent under $\wtilde\mu$.
\item  For $v\in V\setminus W$, the law of $\oD_v$ under $\wtilde \mu$ is $\mu_{2\d}$.  
\item For $w\in W$, the law $\mu_w$ of 
$\oD_w$ is that of $\mu_{2\d}$ skewed by the Radon--Nikodym factor 
$2^{1\{\od(w)=0\}}$, which is to say that
\begin{align}
\mu_w(\oD_w \in H) 
&= \frac1{\a_w}\Bigl[2\mu_{2\d}(\oD_w\in H,\,\od(w)=0) \label{ihp15}\\
&\hskip3cm +
\mu_{2\d}(\oD_w\in H,\,\od(w)\ge 1)\bigr],
\nonumber
\end{align}
for appropriate sets $H$, where 
$$
\a_w=\mu_{2\d}(\od(w)=0)+1.
$$
\end{letlist}

Recall the set $S=A\cup G\cup V(B)$ of switching points.
By (a) above, 
\begin{align}
\wtilde\mu(\partial\Psi(B,G,\oD)=A)&=
\wtilde\mu(\forall v,k:\, |S\cap \ol{J^v_k}|\mbox{ is even})\label{ihp11}\\
&= \prod_{v\in V} \wtilde\mu(\forall k:\, |S\cap \ol{J^v_k}|\mbox{ is even}).
\nonumber
\end{align}
We claim that
\begin{equation}
\wtilde\mu(\forall k:\, |S\cap \ol{J^v_k}|\mbox{ is even})= 
s_2(v)M_{B,G}\Bigl(\exp\bigl\{2\d\bigl(\ev(\psi^A)\cap(v\times K_v)\bigr)\bigr\}\Bigr),
\label{ihp12}
\end{equation}
where $M_{B,G}$ is as before. Recall that $M_{B,G}$ is a product measure.
Once \eqref{ihp12} is proved, \eqref{ihp13} follows by \eqref{def-wt} and \eqref{ihp11}.

For $v\in V\sm W$, the restriction of $\psi^A$ to $v\times K_v$
is determined given $B$
and $G$, whence by (b) above, and the remark prior to \eqref{ihp7},
\begin{align}
\wtilde\mu(\forall k:\, |S\cap \ol{J^v_k}|\mbox{ is even})&=
\mu_{2\d}(\forall k:\, |S\cap \ol{J^v_k}|\mbox{ is even})
\label{ihp14}\\
&= \exp\bigl\{-2\d\bigl(\odd(\psi^A)\cap(v\times K_v)\bigr)\bigr\}.
\nonumber
\end{align}
Equation \eqref{ihp12} follows with $s_2(v) =\exp\{-2\d(v\times K_v)\}$. 

For $w\in W$, by \eqref{ihp15},
\begin{align*}
&\wtilde\mu(\forall k:\, |S\cap \ol{J^w_k}|\mbox{ is even})\\
&\hskip1cm =\frac1{\a_w}\Bigl[2\mu_{2\d}(\oD_w=\es)
+\mu_{2\d}(\oD_w\ne\es,\,\forall k:\, |S\cap \ol{J^w_k}|\mbox{ is even})\Bigr]\\
&\hskip1cm=\frac1{\a_w}\Bigl[\mu_{2\d}(\oD_w=\es)+
\mu_{2\d}(\forall k:\, |S\cap \ol{J^w_k}|\mbox{ is even})\Bigr].
\end{align*}
Let $\psi=\psi^A(B,G)$ be a valid colouring with $\psi(w,0) = \even$.  (If
$(w,0) \in A$, we take $\psi(w,0+)=\even$.)
The colouring $\ol\psi$, obtained from $\psi$ by flipping
all colours on $w \times K_w$, is valid also. 
Taking into account the periodic boundary condition,
\begin{align*}
&\mu_{2\d}(\forall k:\,|S\cap \ol{J^w_k}|\mbox{ is even})\\
&\quad =\mu_{2\d}\bigl(\{\oD_w\cap\odd(\psi)=\es\}\cup
\{\oD_w\cap\ev(\psi)=\es\}\bigr)\\
&\quad =\mu_{2\d}(\oD_w\cap\odd(\psi)=\es)+
\mu_{2\d}(\oD_w\cap\ev(\psi)=\es)
-\mu_{2\d}(\oD_w=\es),
\end{align*}
whence
\begin{align}
&\a_w \wtilde\mu(\forall k:\,|S\cap \ol{J^w_k}|\mbox{ is even})\\
&\hskip1cm =\mu_{2\d}(\oD_w\cap\odd(\psi)=\es)
+\mu_{2\d}(\oD_w\cap\ev(\psi)=\es)\nonumber\\
&\hskip1cm =2M_{B,G}\Bigl(\exp\bigl\{-2\d\bigl(\odd(\psi^A)\cap(w\times K_w)\bigr)\bigr\}\Bigr),
\nonumber
\end{align}
since $\odd(\psi^A) = \odd(\psi)$ with $M_{B,G}$-probability $\frac12$, and
equals $\ev(\psi)$ otherwise.
This proves \eqref{ihp12} with $s_2(w) = 2\exp\{-2\d(w\times K_w)\}/\a_w$. 
\end{proof}

By keeping track of the constants in the above proof, we arrive at the
following statement, which will be useful later.

\begin{lemma}\label{Z'}
The partition function $Z'=Z_K'$ of \eqref{o12} satisfies
$$
Z' = 2^N e^{\l(F)+\g(K)-\d(K)} E(\pd\psi^\es),
$$
where  $N=\sum_{v\in V}m(v)$
is the total number of intervals comprising $K$.
\end{lemma}

\subsection{The backbone}

The concept of the backbone is key to the analysis of \cite{abf}, and
its definition there has a certain complexity. The corresponding
definition is easier in the current setting, because of the fact that
bridges, deaths, and sources have (almost surely) no common point. 

We construct a total order on $K$ by: first ordering the
vertices of $L$, and then using the natural order on $[0,\b)$.
Let $A\subseteq \ol K$ 
be a finite set of sources, and let $B\in\cB$, $G \in \cF$.
Let $\psi$ be a valid colouring.  We will
define a sequence of directed odd paths called the \emph{backbone} and denoted
$\xi=\xi(\psi)$.  Suppose $A=(a_1,a_2,\dotsc,a_n)$ in the above ordering.  
Starting at $a_1$, follow the odd interval (in $\psi$) until you
reach an element of $S=A\cup G \cup V(B)$.  If the first
such point thus encountered is the endpoint of a
bridge, cross it, and continue along the odd interval;  continue likewise
until we first reach a point $t_1\in A\cup G$, at which point we stop.
Note, by the validity of $\psi$, that $a_1\ne t_1$.
The odd path thus traversed is denoted $\zeta^1$;  
we take $\zeta^1$ to be closed (when viewed as a subset of $\ZZ^d\times\RR$).  Repeat
the same procedure with $A$ replaced by $A\setminus\{a_1,t_1\}$, and iterate
until no sources remain.  The
resulting (unordered) set of paths $\xi=(\zeta^1,\dotsc,\zeta^k)$ is called the \emph{backbone} of
$\psi$.  The backbone will also be denoted at times as
$\xi=\zeta^1\circ\dotsb\circ\zeta^k$.  We define $\xi(\#)=\es$.
Note that, apart from the backbone, the remaining odd segments of
$\psi$ form disjoint self-avoiding cycles (or  `eddies').  
Unlike the discrete setting of \cite{abf}, there is a (a.s.) unique way of specifying
the backbone from knowledge of $A$, $B$, $G$ and the valid colouring $\psi$.
See Figure~\ref{backbone_fig}.

The backbone contains all the sources $A$ as endpoints, and the configuration outside $\xi$
may be any sourceless configuration.  Moreover, since $\xi$ is entirely odd,
it does not contribute to the weight $\partial\psi$ in \eqref{def-wt}.  
It follows, using properties of Poisson processes,  that the
conditional expectation $E(\partial\psi^A\mid\xi)$ equals the expected weight
of any sourceless colouring of $K\sm\xi$, which is to say that,
with $\xi:= \xi(\psi^A)$,
\begin{equation}\label{backb_cond_eq}
E(\partial\psi^A\mid\xi)=E_{K\setminus\xi}(\partial\psi^\es)
=: Z_{K\setminus\xi}.
\end{equation}
Cf.\ \eqref{o12} and \eqref{ihp8}, and recall Remark \ref{rem-as}.
We abbreviate $Z_K$ to $Z$, and recall
from Lemma \ref{Z'} that the $Z_R$ differ from the partition
functions $Z_R'$ by certain multiplicative constants.

\begin{figure}[tbp]
\includegraphics{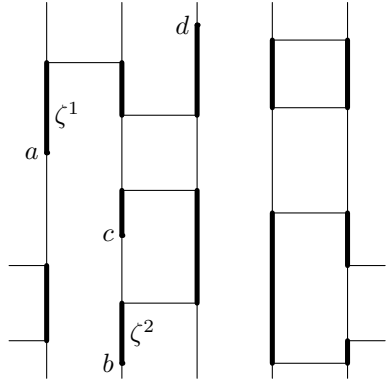}
\caption{A valid colouring configuration $\psi$ with sources $A=\{a,b,c,d\}$, and its backbone 
$\xi=\zeta^1\circ\zeta^2$.  Note that, in this illustration, bridges protruding from
the sides `wrap around', and that there are no ghost-bonds.}
\label{backbone_fig}
\end{figure}

Let $\Xi$ be the set of all possible backbones as $A$, $B$, and $G$ vary, 
regarded as sequences of
directed paths in $K$; these paths may, if required, be ordered by their
starting points.   
For a source-set $A \subseteq\ol K$ and a backbone $\nu\in \Xi$, we write
$A\sim\nu$ if there exists $B \in \cB$ and $G \in \cF$ such
that $M_{B,G}(\xi(\psi^A)=\nu)>0$.
We define the \emph{weight}
$\wt^A(\nu)$ by 
\begin{equation}
\wt^A(\nu) = \wt^A_K(\nu):=
\begin{cases} \dfrac{Z_{K\sm \nu}}{Z} &\text{if } A \sim \nu,\\
0 &\text{otherwise}.
\end{cases}
\label{ihp16}
\end{equation}
By \eqref{backb_cond_eq} and Theorem \ref{rcr_thm},
with $\xi=\xi(\psi^A)$,
\begin{equation}\label{backbone_rep_eq}
E(\wt^A(\xi))=\frac{E(E(\pd\psi^A\mid\xi))}{Z} = 
\frac{E(\pd\psi^A)}{E(\pd\psi^\es)} =\el\s_A\er.
\end{equation}

For $\nu^1,\nu^2\in \Xi$ with $\nu^1\cap \nu^2=\es$ (that is, no point
lies in paths of both $\nu^1$ and $\nu^2$), 
we write $\nu^1\circ \nu^2$ for the element of $\Xi$
comprising the union of $\nu^1$ and $\nu^2$.

Let $\nu =\zeta^1\circ\dotsb\circ\zeta^k\in\Xi$ where $k \ge 1$.  If $\zeta^i$
has starting point $a_i$ 
and endpoint $b_i$, we write $\zeta^i:a_i\rightarrow b_i$, and also
$\nu:a_1\rightarrow b_1,\dotsc,a_k\rightarrow b_k$.  
If $b_i \in G$, we write $\zeta^i: a_i \rightarrow \Gh$.
There is a natural way to `cut' $\nu$ at points $x$ lying on
$\zeta^i$, say, where $x\ne a_i, b_i$: let $\bar\nu^1=\bar\nu^1(\nu,x)=
\zeta^1\circ\cdots\circ\zeta^{i-1}\circ \zeta^i_{\le x}$ and 
$\bar\nu^2=\bar\nu^2(\nu,x)=\zeta^i_{\ge x}\circ\zeta^{i+1}\circ\dots\circ\zeta^k$,
where $\zeta^i_{\le x}$ (\resp, $\zeta^i_{\ge x})$ is the closed sub-path of
$\zeta^i$ from $a_i$ to $x$ (\resp, $x$ to $b_i$).
We express this decomposition as $\nu=\bar\nu^1\circ\bar\nu^2$ where, this
time, each $\bar\nu^i$ may comprise a number of disjoint paths.
The notation $\ol\nu$ will be used only in a situation where there has been a
cut. 

We note two special cases. If $A=\{a\}$, then necessarily
$\xi(\psi^A):a\rightarrow \Gh$, so
\begin{equation}
\el\s_a\er=E\bigl(\wt^a(\xi)\cdot1\{\xi:a\rightarrow \Gh\}\bigr).
\label{special1}
\end{equation}
If $A=\{a,b\}$ where $a<b$ in the ordering of $K$, then 
\begin{equation}
\el\s_a\s_b\er=E\bigl(\wt^{ab}(\xi)\cdot1\{\xi:a\rightarrow b\}\bigr)
+E\bigl(\wt^{ab}(\xi)\cdot1\{\xi:a\rightarrow \Gh,\,b\rightarrow \Gh\}\bigr).
\label{special2}
\end{equation}
The last term equals $0$ when $\g\equiv0$.

Finally, here is a lemma for computing the weight of $\nu$ in terms of its
constituent parts. 
The claim of the lemma is, as usual, valid only `almost surely'.

\begin{lemma}\label{backb2}
(a) Let $\nu^1,\nu^2 \in \Xi$ be disjoint, and $\nu=\nu^1\circ\nu^2$, 
$A\sim\nu$.  Writing $A^i=A\cap\nu^i$,
we have that 
\begin{equation}
\wt^A(\nu)=\wt^{A^1}(\nu^1)\wt^{A^2}_{K\setminus\nu^1}(\nu^2).
\end{equation}
(b) Let $\nu = \ol\nu^1\circ \ol\nu^2$ be a cut of the backbone $\nu$ at
the point $x$, and $A \sim \nu$. Then
\begin{equation}
\wt^A(\nu)=\wt^{B^1}(\ol\nu^1)\wt^{B^2}_{K\setminus\ol\nu^1}(\ol\nu^2).
\end{equation}
where $B^i=A^i\cup\{x\}$.
\end{lemma}

\begin{proof}
By \eqref{ihp16}, the first claim is equivalent to
\begin{equation}
\frac{Z_{K\setminus\nu}}{Z}1\{A\sim\nu\}=
\frac{Z_{K\setminus\nu^1}}{Z}1\{A^1\sim\nu^1\}
\frac{Z_{K\setminus(\nu^1\cup\nu^2)}}{Z_{K\setminus\nu^1}}1\{A^2\sim\nu^2\}.
\end{equation}
The right side vanishes if and only if the left side vanishes. When
both sides are non-zero, their equality follows from the fact that
$Z_{K\setminus\nu}=Z_{K\setminus(\nu^1\cup\nu^2)}$. The second claim
follows similarly, on adding $x$ to the set of sources.
\end{proof}

\section{The switching lemma}\label{sw_sec}

We state and prove next the
principal tool in the random-parity representation, namely
the so-called `switching lemma'. In brief, this
allows us to take two independent colourings, with different sources, and to
`switch' the sources from one to the other in a measure-preserving way.
In so doing,  
the backbone will generally change.  In order to
preserve the measure, the \emph{connectivities}
inherent in the backbone must be retained. We begin by defining
two notions of connectivity in colourings. We work throughout this
section in the general set-up of Section \ref{ssec-col}.

\subsection{Connectivity and switching}\label{ssec-switching}

Let $B\in \cB$, $G \in \cF$, 
let $A \subseteq\ol K$ be a finite set of sources, and write $\psi^A=\psi^A(B,G)$ for the colouring
given in the last section.  In what follows we think of the ghost-bonds
as bridges to the ghost-site $\Gh$. 

Let 
$x,y\in K^\Gh:= K \cup\{\Gh\}$. A \emph{path} from $x$ to $y$ in the configuration $(B,G)$ is a 
self-avoiding path with endpoints $x$, $y$,
traversing intervals of $K^\Gh$, and possibly bridges in $B$ and/or ghost-bonds 
joining $G$ to $\Gh$.  
Similarly, a \emph{cycle} is a self-avoiding cycle in the
above graph. A \emph{route} is a path or a cycle. A route containing no ghost-bonds is called a
\emph{lattice-route}.  A route is called \emph{odd} (in the colouring $\psi^A$)
if $\psi^A$, when restricted to the route, takes only the value `odd'. 
The failed colouring $\psi^A=\#$ is deemed to contain no odd
routes.  

Let $B_1,B_2\in\cB$, $G_1,G_2\in\cF$, and
let $\psi_1^A=\psi_1^A(B_1,G_1)$ and $\psi_2^B=\psi_2^B(B_2,G_2)$
be the associated colourings.
Let $\D$ be an auxiliary Poisson process on $K$, with
intensity function $4\d(\cdot)$, that is independent of all other random
variables so far. We call points of
$\D$ \emph{cuts}.  A route of $(B_1\cup B_2, G_1\cup G_2)$
is said to be \emph{open} in the triple
$(\psi_1^A,\psi_2^B,\D)$ if it includes no 
sub-interval of $\ev(\psi_1^A)\cap\ev(\psi_2^B)$ containing
one or more  elements of $\D$.
In other words, the cuts break paths, but only when they
belong to intervals labelled `even' in \emph{both} colourings.  See Figure~\ref{connectivity_fig}.
In particular, if there is an odd path $\pi$ from $x$
to $y$ in $\psi_1^A$, then $\pi$ constitutes an open path in
$(\psi_1^A,\psi_2^B,\D)$ irrespective of $\psi_2^B$ and $\D$.  We let
\begin{equation}
\{x\lra y\mbox{ in }\psi_1^A,\psi_2^B,\D\}
\end{equation}
be the event that there exists an open path from $x$ to $y$ in 
$(\psi_1^A,\psi_2^B,\D)$. We may abbreviate this to
$\{x\lra y\}$ when there is no ambiguity. 

\begin{figure}[tbp]
\includegraphics{sharptransition.7}
\hspace{1.3cm} 
\includegraphics{sharptransition.8} 
\hspace{1.3cm} 
\includegraphics{sharptransition.9} 
\caption{Connectivity in pairs of colourings.
\emph{Left}:  $\psi_1^{ac}$.  \emph{Middle}:  $\psi_2^\es$.
\emph{Right}:  the triple $\psi_1^{ac},\psi_2^\es,\D$. 
Crosses are elements of $\D$ and grey lines are
where either $\psi_1^{ac}$ or $\psi_2^\es$ is odd.  In
$(\psi_1^{ac},\psi_2^\es,\D)$ the following connectivities hold:
$a\nlra b$, $a\lra c$, $a\lra d$,
$b\nlra c$, $b\nlra d$,
$c\lra d$.  The dotted line marks $\pi$, one of the open paths from
$a$ to $c$.} 
\label{connectivity_fig}
\end{figure}

There is an analogy between open paths in the above construction
and the notion of connectivity in the
random-current representation of the discrete Ising model. Points
labelled `odd' or
`even' above may be considered as collections of infinitesimal parallel edges,
being odd or even in number, respectively.  If a point is `even', the
corresponding
number of edges may be $2,4,6,\dotsc$ \emph{or} it may be 0;  in the `union'
of $\psi_1^A$ and $\psi_2^B$, connectivity is broken at a point if and only if
both the corresponding numbers equal 0.  It turns out that the correct law for
the set of such points is that of $\D$.

Here is some notation.  For any finite sequence $(a,b,c,\dots)$ of elements in $K$, 
the string  
$abc\dotsc$ will denote the subset of elements
that appear an odd number of times in the sequence.  
If $A\subseteq \ol K$ 
is a finite source-set with odd
cardinality, then for any pair $(B,G)$ for which there exists a valid 
colouring $\psi^A(B,G)$, the
number of ghost-bonds must be odd.  Thinking of these as bridges to $\Gh$, $\Gh$ may thus 
be viewed as an element of $A$, and we make the following remark.

\begin{remark}\label{gGremark}
For a source-set $A \subseteq \ol K$ with $|A|$ odd,
we shall use the expressions $\psi^A$ and $\psi^{A\cup\{\Gh\}}$ interchangeably.
\end{remark}
 
We call a function $F$, acting on $(\psi_1^A,\psi_2^B,\D)$, a
\emph{connectivity function} if it depends only on the connectivity
properties using open paths of $(\psi_1^A,\psi_2^B,\D)$, that is, the value
of $F$ depends only on the set $\{(x,y)\in (K^\Gh)^2: x \lra y\}$.
In the following, $E$ denotes expectation with respect to $d\mu_\l\, d\mu_\g\, dM_{B,G}\,dP$
where $P$ is the law of $\Delta$.

\begin{theorem}[Switching lemma]\label{sl}
Let $F$ be a connectivity function and $A,B\subseteq \ol K$ finite source-sets.
For $x,y\in \ol K\cup\{\Gh\}$ such that $A\sd xy$ and $B\sd xy$ are source-sets,
\begin{align}
\label{sw_eq_1}
&E\bigl(\partial\psi_1^A\partial\psi_2^B\cdot F(\psi_1^A,\psi_2^B,\D)
\cdot 1\{x\lra y\mbox{ in }\psi_1^A,\psi_2^B,\D\}\bigr)\\
&\hskip1cm =E\Big(\partial\psi_1^{A\sd xy}\partial\psi_2^{B\sd xy}
\cdot F(\psi_1^{A\sd xy},\psi_2^{B\sd xy},\D)\cdot \nonumber\\ 
&\hskip4cm \cdot 1\{x\lra y\mbox{ in }\psi_1^{A\sd xy},\psi_2^{B\sd xy},\D\}\Big).
\nonumber
\end{align}
In particular,
\begin{equation}\label{sw_eq_2}
E(\partial\psi_1^{xy}\partial\psi_2^B)
=E\bigl(\partial\psi_1^\es\partial\psi_2^{B\sd xy}
\cdot1\{x\lra y\mbox{ in }
\psi_1^\es,\psi_2^{B\sd xy},\D\}\bigr).
\end{equation}
\end{theorem}

\begin{proof}
Equation \eqref{sw_eq_2} follows from \eqref{sw_eq_1} with $A=\{x,y\}$
and $F\equiv 1$, and so it suffices to prove \eqref{sw_eq_1}.
This is trivial if $x=y$, and we assume henceforth that $x\neq y$.
Recall that $W=\{v\in V:K_v=\SS\}$ and $|W|=r$.  

We prove \eqref{sw_eq_1} first for the special case when $F\equiv 1$, that is,
\begin{multline}\label{sw_eq_3}
E\bigl(\partial\psi_1^A\partial\psi_2^B
\cdot 1\{x\lra y\mbox{ in }\psi_1^A,\psi_2^B,\D\}\bigr)\\
=E\bigl(\partial\psi_1^{A\sd xy}\partial\psi_2^{B\sd xy}\cdot
1\{x\lra y\mbox{ in }\psi_1^{A\sd xy},\psi_2^{B\sd xy},\D\}\bigr),
\end{multline}
and this will follow by conditioning on the pair $Q=(B_1\cup B_2,G_1\cup G_2)$.
 
Let $Q\in \cB\times \cF$ be given. Conditional on 
$Q$, the law of $(\psi_1^A,\psi_2^B)$ is given as follows.
First, we
allocate each bridge and each ghost-bond to either $\psi_1^A$ or
$\psi_2^B$ with equal probability (independently of one another).
If $W \ne \es$, then we must also allocate (uniform) random colours
to the points $(w,0)$, $w \in W$, for each of $\psi_1^A$, $\psi_2^B$.
If $(w,0)$ is itself a source, we work with $(w,0+)$.
(Recall that the pair $(B',G')$ may be reconstructed from
knowledge of a valid colouring $\psi^{A'}(B',G')$.)
There are $2^{|Q|+2r}$ possible outcomes of the above choices, and each is
equally likely.

The process $\D$ is independent of
all random variables used above.  Therefore, the conditional expectation,
given $Q$, of the random variable on the left side of \eqref{sw_eq_3} equals  
\begin{equation}\label{sw_cond_eq}
\frac{1}{2^{|Q|+2r}}\sum _{\Qab}\partial Q_1\partial Q_2\,
P(x\lra y\mbox{ in }Q_1,Q_2,\D),
\end{equation}
where the sum is over the set $\Qab=\Qab(Q)$ of all possible pairs $(Q_1,Q_2)$ of
values of $(\psi_1^A,\psi_2^B)$. 
The measure $P$ is that of $\D$. 

We shall define an invertible (and therefore measure-preserving) map 
from $\Qab$ to $\Qabxy$. Let $\pi$ be a path of $Q$ with endpoints $x$ and $y$
(if such a path $\pi$ exists), and let $f_\pi:\Qab\to\Qabxy$
be given as follows.
Let $(Q_1,Q_2)\in\Qab$, say $Q_1=Q_1^A(B_1,G_1)$
and $Q_2=Q_2^B(B_2,G_2)$ where $Q=(B_1\cup B_2, G_1\cup G_2)$.
For $i=1,2$, let $B_i'$ (\resp, $G_i'$) be the set of bridges (\resp, ghost-bonds)
in $Q$ lying in exactly one
of $B_i$, $\pi$ (\resp, $G_i$, $\pi$). Otherwise expressed, $(B_i',G_i')$ is obtained
from $(B_i,G_i)$ by adding the bridges/ghost-bonds of $\pi$ `modulo 2'.
Note that $(B_1'\cup B_2', G_1' \cup G_2') = Q$.

If $W = \es$, we let $R_1=R_1^{A \sd xy}$ (\resp, $R_2^{B\sd xy}$) be the
unique valid colouring of $(B_1',G_1')$ with sources $A\sd xy$
(\resp, $(B_2',G_2')$ with sources $B\sd xy$), so
$R_1=\psi^{A\sd xy}(B_1',G_1')$, and similarly for $R_2$.
When $W\ne \es$ and $i=1,2$, we choose the colours of the $(w,0)$, $w\in W$, 
(or $(w,0+)$ if $(w,0)$ is a source)
 in $R_i$ in such a way
that $R_i \equiv Q_i$ on $K\sm\pi$. 

It is easily seen that the map $f_\pi:(Q_1,Q_2) \mapsto (R_1,R_2)$
is invertible, indeed its inverse is given by the same mechanism.
See Figure~\ref{connectivity_after_switch_fig}.

\begin{figure}[tbp]
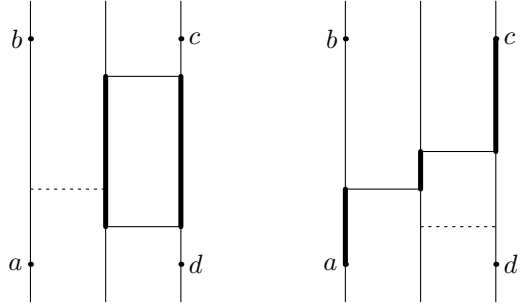

\includegraphics{sharptransition.10}
\hspace{1.3cm} 
\includegraphics{sharptransition.11} 
\caption{Switched configurations.  Taking $Q_1^{ac}$, $Q_2^\es$ and
  $\pi$ to be $\psi_1^{ac}$, $\psi_2^\es$ and $\pi$ of
  Figure~\ref{connectivity_fig}, respectively, this figure illustrates the
  `switched' configurations $R_1^\es$ and $R_2^{ac}$ (left and
  right, respectively).}
\label{connectivity_after_switch_fig}
\end{figure}

By \eqref{def-wt}, 
\begin{equation}
\partial Q_1\partial Q_2=
\exp\bigl\{2\d(\ev(Q_1))+2\d(\ev(Q_2))\bigr\}.
\label{o14}
\end{equation}
Now,
\begin{align}\label{o15}
\d(\ev(Q_i))
&=\d(\ev(Q_i)\cap\pi)+\d(\ev(Q_i)\setminus\pi)\\
&=\d(\ev(Q_i)\cap\pi)+\d(\ev(R_i)\setminus\pi), 
\nonumber
\end{align}
and
\begin{align*}
&\d(\ev(Q_1)\cap\pi)+\d(\ev(Q_2)\cap\pi) - 2\d\bigl(\ev(Q_1)\cap\ev(Q_2)\cap\pi\bigr) \\
&\hskip1cm =\d\bigl(\ev(Q_1)\cap\odd(Q_2)\cap\pi\bigr)+\d\bigl(\odd(Q_1)\cap\ev(Q_2)\cap\pi\bigr)\\
&\hskip1cm=\d\bigl(\odd(R_1)\cap\ev(R_2)\cap\pi\bigr)+\d\bigl(\ev(R_1)\cap\odd(R_2)\cap\pi\bigr)\\
&\hskip1cm =\d(\ev(R_1)\cap\pi)+\d(\ev(R_2)\cap\pi) -2\d\bigl(\ev(R_1)\cap\ev(R_2)\cap\pi\bigr),
\end{align*}
whence, by \eqref{o14}--\eqref{o15},
\begin{align}\label{switched_weights_eq}
\partial Q_1\partial Q_2=
\partial R_1\partial R_2&
\exp\bigl\{-4\d\bigl(\ev(R_1)\cap\ev(R_2)\cap\pi\bigr)\bigr\}\\ 
&\times \exp\bigl\{4\d\bigl(\ev(Q_1)\cap\ev(Q_2)\cap\pi\bigr)\bigr\}.
\nonumber
\end{align}

The next step is to choose a suitable path $\pi$.
Consider the final term in \eqref{sw_cond_eq}, namely 
\begin{equation}
P(x\lra y\mbox{ in }Q_1,Q_2,\D).
\end{equation}
There are finitely many paths in $Q$ from $x$ to $y$, let these paths be
$\pi_1,\pi_2,\dotsc,\pi_n$. Let 
$\cO_k=\cO_k(Q_1,Q_2,\D)$ be the event that
$\pi_k$ is the earliest such path that is open in $(Q_1,Q_2,\D)$.  Then 
\begin{align}\label{path_prob_eq}
&\hskip-1cm P(x\lra y\mbox{ in }Q_1,Q_2,\D)\\
&=\sum_{k=1}^n P(\cO_k)\nonumber\\
&=\sum_{k=1}^n P\bigl(\D\cap[\ev(Q_1)\cap\ev(Q_2)\cap \pi_k] = \es\bigr)
P(\wtilde\cO_k)\nonumber\\
&=\sum_{k=1}^n 
\exp\bigl\{-4\d\bigl(\ev(Q_1)\cap\ev(Q_2)\cap\pi_k\bigr)\bigr\}
P(\wtilde\cO_k),
\nonumber
\end{align}
where $\wtilde\cO_k = \wtilde\cO_k(Q_1,Q_2,\D)$ is the 
event (that is, subset of $\cF$) that each of  
$\pi_1,\dotsc,\pi_{k-1}$ is rendered non-open in $(Q_1,Q_2,\D)$
through the presence of elements of $\D$ lying in $K\sm\pi_k$.
In the second line of \eqref{path_prob_eq}, we have used the independence
of $\D \cap \pi_k$ and $\D \cap(K\sm \pi_k)$. 

Let $(R_1^k,R_2^k) =f_{\pi_k}(Q_1,Q_2)$.  
Since $R_i^k \equiv Q_i$ on $K \sm \pi_k$, we have that
$\wtilde\cO_k(Q_1,Q_2,\D) = \wtilde\cO_k(R_1^k,R_2^k,\D)$. 
By
\eqref{switched_weights_eq} and~\eqref{path_prob_eq}, the
summand in \eqref{sw_cond_eq} equals
\begin{align*}
& \sum_{k=1}^n \partial Q_1\partial Q_2
\exp\bigl\{-4\d\bigl(\ev(Q_1)\cap\ev(Q_2)\cap\pi_k\bigr)\bigr\}P(\wtilde\cO_k)\\
&\hskip1cm=\sum_{k=1}^n \partial R_1^k\partial R_2^k
\exp\bigl\{-4\d\bigl(\ev(R_1^k)\cap\ev(R_2^k)\cap\pi_k\bigr)\bigr\}
P(\wtilde\cO_k)\\
&\hskip1cm=\sum_{k=1}^n \partial R_1^k
\partial R_2^k \,P(\cO_k(R_1^k,R_2^k,\D)).
\end{align*}

Summing the above over $\Qab$, and remembering that each $f_{\pi_k}$
is
a bijection between $\Qab$ and $\Qabxy$, \eqref{sw_cond_eq} becomes
\begin{align*}
\frac1{2^{|Q|+2r}} \sum_{k=1}^n\, 
& \sum_{(R_1,R_2)\in\Qabxy}\partial R_1\partial R_2\,
P(\cO_k(R_1,R_2,\D))\\
&= \frac1{2^{|Q|+2r}} \sum_{\Qabxy}\partial R_1\partial R_2\,
P(x\lra y\mbox{ in }R_1,R_2,\D).
\end{align*}
By the argument leading to \eqref{sw_cond_eq}, this equals the right
side of \eqref{sw_eq_3}, and the claim is proved
when $F\equiv 1$.  

Consider now the case of general connectivity functions $F$ in \eqref{sw_eq_1}.
In~\eqref{sw_cond_eq}, the factor 
$P(x\lra y\mbox{ in }Q_1,Q_2,\D)$ is replaced by 
$$
P\bigl(F(Q_1,Q_2,\D)\cdot 1\{x\lra y\mbox{ in }Q_1,Q_2,\D\}\bigr),
$$
where $P$ denotes expectation with respect to $\D$.
In the calculation~\eqref{path_prob_eq}, we use the fact that
$$
P(F\cdot 1_{\cO_k})=P(F\mid\cO_k)P(\cO_k)
$$
and we deal with the factor $P(\cO_k)$ as before.  The result follows on noting
that, for each $k$,
$$
P\bigl(F(Q_1,Q_2,\D)\bigmid \cO_k(Q_1,Q_2,\D)\bigr)=
P\bigl(F(R_1^k,R_2^k,\D)\bigmid \cO_k(R_1^k,R_2^k,\D)\bigr).
$$
This holds because: (i) the configurations $(Q_1,Q_2,\D)$ and $(R_1^k,R_2^k,\D)$ are
identical off $\pi_k$, and (ii) in each, all points along $\pi_k$ are
connected. Thus the connectivities are identical in the two configurations.
\end{proof}

\subsection{Applications of switching}\label{sw_appl_sec}

In this section are presented a number of inequalities and identities
proved using the random-parity representation and 
the switching lemma.  With some exceptions (most notably~\eqref{dd_bound_eq}) the
proofs are adaptations of 
the proofs for the discrete Ising model that may be found in \cite{abf,grimmett_RCM}. 

For functions $f,g: K \to \RR$, we write $f\le g$ if $f(x) \le g(x)$ for all $x \in K$.

\begin{lemma}[\gks\ inequality]\label{gks_lem}
Let $A,B\subseteq \ol K$ be finite sets of sources, not necessarily disjoint.  Then
\begin{equation}\label{gks_1_eq}
\el\s_A\er\geq 0,
\end{equation}
and
\begin{equation}\label{gks_2_eq}
\el\s_A;\s_B\er := \el\s_A\s_B\er - \el \s_A\er\el\s_B\er \geq 0.
\end{equation}
\end{lemma}

\begin{lemma}\label{cor_mon_lem}
Let $A\subseteq \ol K$ be a finite set of sources.  Then $\el\s_A\er$ is increasing in
$\l$ and $\g$ and decreasing in $\d$.  Moreover, if $R\subseteq K$ is measurable,
\begin{equation}\label{cor_mon_eq}
\el\s_A\er_{K\setminus R}\leq \el\s_A\er_K.
\end{equation}
\end{lemma}

We interpret $\el\s_A\er_{K\sm R}$ as $0$ when $A$ is not a source-set for $K\sm R$.

Lemmas~\ref{gks_lem} and~\ref{cor_mon_lem} may be shown using conventional
inequalities of spin-correlation-type.
They may be proved more easily
using the \fkg-inequality for the associated \rc\ model (using, for
example, the methods of \cite{grimmett_gks}).  We omit these proofs, full details
of which may be found in~\cite{bjo_phd}.

For $R\subseteq K$ a finite union of intervals, let
\begin{equation*}
\wtilde R:=\{(uv,t)\in F: \mbox{either } (u,t)\in R\mbox{ or }(v,t)\in R\mbox{ or both}\}.
\end{equation*}
Recall that $W=W(K)=\{v\in V: K_v=\SS\}$, and $N=N(K)$ is the total number of 
(maximal) intervals constituting $K$.

\begin{lemma}\label{rw_mon_lem}
Let $R\subseteq K$ be a finite union of intervals, 
and let $\nu\in\Xi$ be such that $\nu\cap
R=\es$.  If $A \subseteq \ol{K\sm R}$ is a finite source-set for both $K$ and
$K\sm R$,
and  $A\sim\nu$, then
\begin{equation}
\wt^A(\nu)\leq 2^{r(\nu)-r'(\nu)}\wt^A_{K\setminus R}(\nu),
\end{equation}
where 
\begin{align*}
r(\nu) &= r(\nu,K) := |\{w\in W: \nu\cap (w\times K_w) \ne \es\}|,\\
r'(\nu) &= r(\nu,K\sm R).
\end{align*}
\end{lemma}

\begin{proof}
By \eqref{ihp16} and  Lemma~\ref{Z'},
\begin{align}\label{r0}
\wt^A(\nu)&=\frac{Z_{K\setminus\nu}}{Z_K}\\
&=2^{N(K)-N(K\sm\nu)}
e^{\l(\wtilde\nu)+\g(\nu)-\d(\nu)}\frac{Z'_{K\sm\nu}}{Z'_K}.
\nonumber
\end{align}
We claim that
\begin{equation}\label{r05}
\frac{Z'_{K\sm\nu}}{Z'_K}\leq\frac{Z'_{K\sm(R\cup\nu)}}{Z'_{K\sm R}},
\end{equation}
and the proof of this follows.

Recall the formula~\eqref{o12} for $Z'_K$ in terms of an integral over the
Poisson process $D$.  The set $D$ is the union of
independent Poisson processes $D'$ and $D''$, restricted respectively to $K\sm\nu$ and $\nu$.
We write $P'$ (\resp, $P''$) for the probability measure (and, on occasion,
expectation operator) governing $D'$ (\resp, $D''$).
Let $\S(D')$ denote the set of spin configurations on $K\sm\nu$
that are permitted by $D'$.  By \eqref{o12},
\begin{equation}
\label{Z'_split_eq}
Z'_K= P'\left( \sum_{\s'\in\S(D')}Z_\nu'(\s') 
\exp\left\{\int_{F\sm\wtilde\nu}\l(e)\s'_e\,de+\int_{K\sm\nu}\g(x)\s'_x\,dx\right\}\right),
\end{equation}
where
$$
Z_\nu'(\s') = P''\left(\sum_{\s''\in\wtilde\S(D'')}
\exp\left\{\int_{\wtilde\nu}\l(e)\s_e\,de+\int_{\nu}\g(x)\s_x\,dx\right\}\cdot 1_C(\s')\right)
$$
is the partition function on $\nu$ with boundary condition $\s'$,
and where $\s$, $\wtilde\S(D'')$, and $C=C(D'')$ are given as follows.

The set $D''$ divides $\nu$, in the usual way, into a collection 
$V_\nu(D'')$ of intervals.  From the set of endpoints of
such intervals, we distinguish the subset $\cE$ that: (i) lie in $K$, and (ii) 
are endpoints of some interval of $K\sm \nu$. For $x\in\cE$,
let $\s'_x =\lim_{y\to x} \s'_y$, where the limit is taken over $y\in K\sm\nu$.
Let $\wtilde V_\nu(D'')$ be the subset of $V_\nu(D'')$ containing
those intervals with no endpoint in $\cE$,
and let $\wtilde\S(D'') =\{-1,+1\}^{\wtilde V_\nu(D'')}$.

Let $\s'\in \S(D')$, and 
let $\cI$ be the set of maximal sub-intervals $I$ of $\nu$
having both endpoints in $\cE$, and such that $I \cap D''=\es$.
Let $C=C(D'')$ be the set of $\s'\in\S(D')$ such that, for all $I\in\cI$,
the endpoints of $I$ have equal spins under $\s'$.
Note that
\begin{equation}
\label{o30}
1_C(\s') = \prod_{I\in\cI} \tfrac12(\s'_{x(I)}\s'_{y(I)} + 1),
\end{equation}
where $x(I)$, $y(I)$ denote the endpoints of $I$.

Let $\s''\in \wtilde\S(D'')$. The conjunction $\s$ of $\s'$ and $\s''$ is defined
except on sub-intervals of $\nu$ lying in $V_\nu(D'')\sm \wtilde V_\nu(D'')$. 
On any such sub-interval with exactly one endpoint $x$ in $\cE$, we set
$\s\equiv \s'_x$. On the event $C$, an interval of $\nu$
with both endpoints $x(I)$, $y(I)$ in $\cE$
receives the spin $\s\equiv \s_{x(I)}' = \s_{y(I)}'$.
Thus, $\s\in \S(D'\cup D'')$ is well defined for $\s'\in C$.

By \eqref{Z'_split_eq}, 
\begin{equation*}
\frac{Z'_K}{Z'_{K\sm\nu}}=\el Z'_\nu(\s')\er_{K\sm\nu}.
\end{equation*}
Taking the expectation $\el\cdot\er_{K\sm\nu}$ inside the integral, the last expression becomes
\begin{equation*}
P''\left(\sum_{\s''\in\wtilde\S(D'')}\left\el
\exp\left\{\int_{\wtilde\nu}\l(e)\s_e\,de\right\}
\exp\left\{\int_{\nu}\g(x)\s_x\,dx\right\}
\cdot 1_C(\s')\right\er_{K\sm\nu}\right)
\end{equation*}
The inner expectation may be expressed as a sum over $k,l\geq 0$ 
(with non-negative coefficients) of iterated integrals of the form
\begin{equation}
\frac1{k!}\,\frac1{l!}\,\iint\limits_{\wtilde\nu^k\times\nu^l}\l(\mathbf e)\g(\mathbf x)
\el\s_{e_1}\cdots\s_{e_k}\s_{x_1}\cdots\s_{x_l}\cdot 1_C(\s')\er_{K\sm\nu}
\,d\mathbf e \,d\mathbf x,
\label{o32}
\end{equation}
where we have written $\mathbf e=(e_1,\dotsc,e_k)$, and $\l(\mathbf e)$
for $\l(e_1)\dotsb\l(e_k)$ (and similarly for $\mathbf x$ and $\g(\mathbf x)$). 
We may write
\begin{equation*}
\el\s_{e_1}\cdots\s_{e_k}\s_{x_1}\cdots\s_{x_l}\cdot 1_C\er_{K\sm\nu}
=\el\s'_S\s''_T\cdot 1_C\er_{K\sm\nu}=\s''_T\el\s_S'\cdot 1_C\er_{K\sm\nu},
\end{equation*}
for sets $S\subseteq \ol{K\sm\nu}$, $T\subseteq \nu$
 determined by $e_1,\dotsc,e_k,x_1,\dotsc,x_l$ and $D''$ only.
We now bring the sum over $\s''$ inside the integral of \eqref{o32}.  For $T\neq\es$,
\begin{equation*}
\sum_{\s''\in\wtilde\S(D'')}\s''_T\el\s_S'\cdot 1_C\er_{K\sm\nu}=0,
\end{equation*}
so any non-zero term is of the form
\begin{equation}\label{r1}
\el\s_S'\cdot 1_C\er_{K\sm\nu}.
\end{equation}

By \eqref{o30}, \eqref{r1} 
may be expressed in the form 
\begin{equation}
\sum_{i=1}^s2^{-a_i}\el\s'_{S_i}\er_{K\sm\nu}
\label{o35}
\end{equation}
for appropriate sets $S_i$ and integers $a_i$.  By Lemma~\ref{cor_mon_lem},
\begin{equation*}
\el\s'_{S_i}\er_{K\sm\nu}\geq\el\s'_{S_i}\er_{K\sm(R\cup\nu)}.
\end{equation*}
On working backwards, we obtain \eqref{r05}.

By \eqref{r0}--\eqref{r05},
\begin{equation*}
\wt^A(\nu)\leq 2^U\wt^A_{K\setminus R}(\nu),
\end{equation*}
where 
\begin{align*}
U&=\bigl[N(K)-N(K\sm\nu)\bigr]- \bigl[N(K\sm R)-N(K\sm (R\cup\nu) )\bigr]\\
&=r(\nu)-r'(\nu)
\end{align*}
as required.  
\end{proof}

For distinct $x,y,z\in K^\Gh$, let
\begin{align*}
\el\s_x;\s_y;\s_z\er &:=
\el\s_{xyz}\er -\el\s_{x}\er\el\s_{yz}\er\\
&\hskip1.5cm -\el\s_{y}\er\el\s_{xz}\er
-\el\s_{z}\er\el\s_{xy}\er
+2\el\s_{x}\er\el\s_{y}\er\el\s_{z}\er.
\end{align*}

\begin{lemma}[\ghs\ inequality]\label{ghs_lem}
For distinct $x,y,z\in K^\Gh$, 
\begin{equation}\label{ghs_1_eq}
\el\s_x;\s_y;\s_z\er\leq 0.
\end{equation}
Moreover, $\el\s_x\er$ is concave in $\g$ in the sense that, for
bounded, measurable functions $\g_1,\g_2: K\to\RRp$ satisfying
$\g_1\le\g_2$, and
$\theta\in[0,1]$, 
\begin{equation}
\theta\el\s_x\er_{\g_1}+(1-\theta)\el\s_x\er_{\g_2}\leq
\el\s_x\er_{\theta\g_1+(1-\theta)\g_2}.
\end{equation}
\end{lemma}

\begin{proof}
The proof of this follows very closely the corresponding proof for the
classical Ising model~\cite{ghs}. We include it here because it allows us
to develop the technique of `conditioning on clusters', which will be 
useful later.

We prove~\eqref{ghs_1_eq} via the following more general result.
Let $(B_i,G_i)$, $i=1,2,3$, be independent sets of bridges/ghost-bonds,
and write $\psi_i$, $i=1,2,3$, for corresponding colourings (with sources
to be specified through their superscripts).
We claim that, for any four points $w, x, y, z\in K^\Gh$,
\begin{equation}
\label{ghs_2_eq}
\begin{split}
&E\bigl(\partial\psi_1^\es\partial\psi_2^\es \partial\psi_3^{wxyz}\bigr)-
E\bigl(\partial\psi_1^\es\partial\psi_2^{wz}\partial\psi_3^{xy}\bigr)
\\
&\quad\leq E(\partial\psi_1^\es \partial\psi_2^{wx} \partial\psi_3^{yz})
+E(\partial\psi_1^\es\partial\psi_2^{wy}\partial\psi_3^{xz})
-2E(\partial\psi_1^{wx}\partial\psi_2^{wy}\partial\psi_3^{wz}).
\end{split}
\end{equation}
Inequality \eqref{ghs_1_eq} follows by Theorem \ref{rcr_thm} on letting
$w=\Gh$. 

The left side of \eqref{ghs_2_eq} is
\begin{align*}
&E(\partial\psi_1^\es)\bigl[
E(\partial\psi_2^\es\partial\psi_3^{wxyz})-
E(\partial\psi_2^{wz}\partial\psi_3^{xy})\bigr]\\
&\hskip3cm =
Z\, E\bigl(\partial\psi_2^\es\partial\psi_3^{wxyz}
\cdot1\{w\nlra z\}\bigr),
\end{align*}
by the switching lemma \ref{sl}.  When $\partial\psi_3^{wxyz}$  is
non-zero, parity constraints imply that at least one of  
$\{w\lra x\}\cap \{y\lra z\}$ and $\{w\lra y\}\cap 
\{x\lra z\}$ occurs, but that, in the presence of the indicator function 
they cannot both occur.  Therefore,
\begin{align}
\label{ghs_pf_1_eq}
&E(\partial\psi_2^\es\partial\psi_3^{wxyz}
\cdot1\{w\nlra z\})\\
&\hskip1cm
=E\bigl(\partial\psi_2^\es\partial\psi_3^{wxyz}
\cdot1\{w\nlra z\}
\cdot1\{w\lra x\}\bigr)\nonumber\\
&\hskip3cm +
E\bigl(\partial\psi_2^\es\partial\psi_3^{wxyz}
\cdot1\{w\nlra z\}
\cdot1\{w\lra y\}\bigr).
\nonumber
\end{align}
Consider the first term. By the switching lemma,
\begin{equation}
E\bigl(\partial\psi_2^\es\partial\psi_3^{wxyz}
\cdot1\{w\nlra z\}
\cdot1\{w\lra x\}\bigr)=
E\bigl(\partial\psi_2^{wx}\partial\psi_3^{yz}
\cdot1\{w\nlra z\}\bigr).
\label{o17}
\end{equation}

We next `condition on a cluster'.  Let 
$C_z=C_z(\psi_2^{wx},\psi_3^{yz},\D)$ be the set of all points of
$K$ that are connected by open paths to $z$.  Conditional on 
$C_z$, define new independent colourings
$\mu_2^\es$, $\mu_3^{yz}$ on the domain $M=C_z$.  Similarly,
let $\nu_2^{wx}$, $\nu_3^\es$ be independent colourings on the domain
$N=K\setminus C_z$, that are also independent of the $\mu_i$.  It is
not hard to see that, if $w\nlra z$ in $(\psi_2^{wx},\psi_3^{yz},\D)$, then, conditional on $C_z$,
the law of $\psi_2^{wx}$ equals 
that of the superposition of $\mu_2^\es$ and $\nu_2^{wx}$;
similarly the conditional law of $\psi_3^{yz}$ is the same as that of the
superposition of $\mu_3^{yz}$ and $\nu_3^\es$. 
Therefore, almost surely on the event $\{w \nlra z\}$,
\begin{align}
E(\partial\psi_2^{wx}\partial\psi_3^{yz}\mid C_z)&=
E'(\pd\mu_2^\es)E'(\pd\nu_2^{wx})E'(\pd\mu_3^{yz})E'(\pd\nu_3^\es)\label{o19}\\
&=\el\s_{wx}\er_N E'(\pd\mu_2^\es)E'(\pd\nu_2^\es)
E'(\pd\mu_3^{yz})E'(\pd\nu_3^\es)\nonumber\\
&\leq
\el\s_{wx}\er_KE(\partial\psi_2^\es\partial\psi_3^{yz}\mid C_z),
\nonumber\end{align}
where $E'$ denotes expectation conditional on $C_z$,
and we have used Lemma \ref{cor_mon_lem}. 
Returning to \eqref{ghs_pf_1_eq}--\eqref{o17},
\begin{align*}
&E\bigl(\partial\psi_2^\es\partial\psi_3^{wxyz}
\cdot1\{w\nlra z\}
\cdot1\{w\lra x\}\bigr)\\
&\hskip2cm \leq
\el\s_{wx}\er E(\partial\psi_2^\es\partial\psi_3^{yz}
\cdot1\{w\nlra z\}).
\end{align*}
The other term in \eqref{ghs_pf_1_eq} satisfies the same inequality
with $x$ and $y$ interchanged.
Inequality \eqref{ghs_2_eq} follows on applying the switching lemma to the right sides
of these two last inequalities, and adding them.

The concavity of $\el\s_x\er$ follows from the fact that, if 
\begin{equation}
T=\sum_{k=1}^n a_k1_{A_k} 
\end{equation}
is a step function on $K$ with $a_k\ge 0$ for all $k$,
and $\g(\cdot)=\g_1(\cdot)+\a T(\cdot)$, then 
\begin{equation}
\frac{\partial^2}{\partial \a^2}\el\s_x\er
=\sum_{k,l=1}^na_ka_l\iint_{A_k\times A_l} dy\,dz\, \el\s_x;\s_y;\s_z\er\leq 0.
\end{equation}
Thus, the claim holds whenever $\g_2-\g_1$ is a step function.  The general
claim follows by approximating $\g_2-\g_1$ by step functions, and applyng the
dominated convergence theorem. 
\end{proof}

For the next lemma we assume for simplicity that $\g\equiv 0$ (although
similar results can easily be proved for $\g\not\equiv 0$).  We let $\bar\d\in\RR$ be
an upper bound for $\d$, thus $\d(x)\leq\bar\d<\oo$ for all $x\in K$.
Let $a,b\in K$ be two distinct points.  A closed set $T\subseteq K$
is said to \emph{separate} $a$ from 
$b$ if every lattice path from $a$ to $b$ (whatever the set of bridges) intersects $T$.  
Moreover, if $\eps>0$
and $T$ separates $a$ from $b$, we say that $T$ is an \emph{$\eps$-fat
separating set} if every point in $T$ lies in a closed sub-interval of $T$ of length at
least $\eps$.

\begin{lemma}[Simon inequality]\label{simon_lem}
Let $\g\equiv 0$.
If $\eps>0$ and $T$ is an $\eps$-fat separating set for $a,b\in K$,
\begin{equation}
\el\s_a\s_b\er\leq\frac{1}{\eps}\exp(8\eps\bar\d)
\int_T\el\s_a\s_x\er\el\s_x\s_b\er\, dx.
\end{equation}
\end{lemma}

\begin{proof}
By Theorems \ref{rcr_thm} and \ref{sl},
\begin{equation}
\el\s_a\s_x\er\el\s_x\s_b\er=
\frac{1}{Z^2}E(\partial\psi_1^\es\partial\psi_2^{ab}
\cdot1\{a\lra x\}),
\end{equation}
and, by Fubini's theorem,
\begin{equation}
\int_T\el\s_a\s_x\er\el\s_x\s_b\er\;dx=
\frac{1}{Z^2}E(\partial\psi_1^\es\partial\psi_2^{ab}
\cdot|\what T|),
\end{equation}
where $\what T=\{x\in T:a\lra x\}$ and $|\cdot|$ denotes Lebesgue
measure.  Since $\g\equiv 0$, the backbone $\xi = \xi(\psi_2^{ab})$ 
consists of a single (lattice-) path from $a$ to $b$
passing through $T$.  Let $U$ denote the set of points in $K$
that are separated from $b$ by $T$, and let $X$ be the point at which $\xi$
exits $U$ for the first time.
  Since $T$ is
assumed closed, $X\in T$.
See Figure \ref{simon_fig}.  

\begin{figure}[tbp]
\includegraphics{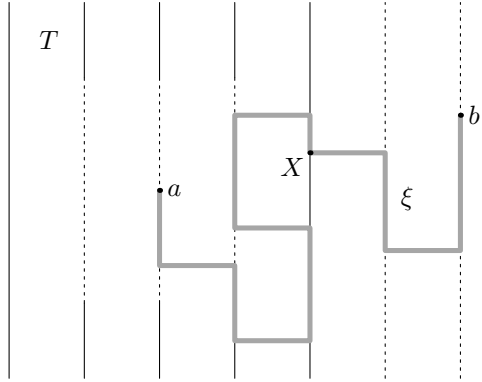}
\caption{The Simon inequality.  The separating set $T$ is drawn with solid
  black lines, and the backbone $\xi$ with a grey line.}
\label{simon_fig}
\end{figure}

For  
$x\in T$, let $A_x$ be the event that there is no element of $\D$ within
the interval of length $2\eps$ centred at $x$.  Thus,
$P(A_x)\ge\exp(-8\eps\bar\d)$.
  On the
event $A_X$, we have that $|\what T|\geq\eps$, whence
\begin{align}\label{simon_eq1_eq}
E(\partial\psi_1^\es\partial\psi_2^{ab}
\cdot|\what T|)&\geq
E(\partial\psi_1^\es\partial\psi_2^{ab}
\cdot|\what T|\cdot1\{A_X\})\\
&\geq\eps
E(\partial\psi_1^\es\partial\psi_2^{ab}
\cdot1\{A_X\}).
\nonumber
\end{align}
Conditional on $X$, the event $A_X$ is independent of $\psi_1^\es$ and
$\psi_2^{ab}$, so that
\begin{equation}
E(\partial\psi_1^\es\partial\psi_2^{ab}
\cdot|\what T|)\geq\eps\exp(-8\eps\bar\d)
E(\partial\psi_1^\es\partial\psi_2^{ab}),
\label{o18}
\end{equation}
and the proof is complete.
\end{proof}

Just as for the classical Ising model, only a small amount of extra work is
required to obtain the following improvement of Lemma~\ref{simon_lem}.

\begin{lemma}[Lieb inequality]\label{lieb_lem}
Under the assumptions of Lemma \ref{simon_lem},
\begin{equation}
\el\s_a\s_b\er\leq\frac{1}{\eps}\exp(8\eps\bar\d)
\int_T\el\s_a\s_x\er_{U}\,\el\s_x\s_b\er\;dx,
\end{equation}
where $\el\cdot\er_{U}$ denotes expectation with respect to the
measure restricted to $U$.
\end{lemma}

\begin{proof}
Let $x \in T$, let $\ol\psi_1^{ax}$ denote a colouring on the restricted region
$U$, and let $\psi_2^{xb}$ denote a colouring on the
full region $K$ as before.  We claim that
\begin{equation}\label{lieb_sw_eq}
E(\partial\ol\psi_1^{ax}\partial\psi_2^{xb})=
E\bigl(\partial\ol\psi_1^\es\partial\psi_2^{ab}
\cdot 1\{a\lra x\mbox{ in }U\}\bigr).
\end{equation}
The use of the letter $E$ is an abuse of notation,
since the $\ol\psi$ are colourings of $U$ only.

Equation \eqref{lieb_sw_eq} may be established using a slight variation in the proof of the
switching lemma.  We follow the proof of that lemma, first
conditioning on the set $Q$ of all bridges and ghost-bonds in the two
colourings taken together, and then allocating them to the colourings $Q_1$ and $Q_2$,
uniformly at random.  We then order the paths $\pi$ of $Q$ from $a$ to
$x$, and add the earliest open path to both $Q_1$ and $Q_2$ `modulo 2'.  There
are two differences here:  firstly, any element of $Q$ that is not
contained in $U$ will be allocated to $Q_2$, and secondly, we only consider paths $\pi$ that lie 
inside $U$.  Subject to these two changes, we follow the argument of
the switching lemma to arrive at~\eqref{lieb_sw_eq}.  

Integrating \eqref{lieb_sw_eq} over $x \in T$,
\begin{equation}
\int_T\el\s_a\s_x\er_{U}\,\el\s_x\s_b\er\;dx=
\frac{1}{Z_{U}Z}E(\partial\ol\psi_1^\es\partial\psi_2^{ab}
\cdot|\what T|),
\end{equation}
where this time 
$\what T=\{x\in T:a\lra x\mbox{ in }U\}$.  
The proof is completed as in \eqref{simon_eq1_eq}--\eqref{o18}.
\end{proof}

For the next lemma we specialize to the situation that is the main focus of
this article, namely the following. Similar results are valid for other lattices and for
summable translation-invariant interactions.

\begin{assumption}\label{periodic_assump}\hspace{1cm}
\begin{itemize}
\item The graph $L=[-n,n]^d \subseteq \ZZ^d$ where $d \ge 1$,
with periodic boundary condition.
\item The parameters $\l$, $\d$, $\g$ are non-negative constants.
\item The set $K_v=\SS$ for every $v\in V$.
\end{itemize}
\end{assumption}
Under the periodic boundary condition, two vertices of $L$ are 
joined by an edge whenever there exists $i\in\{1,2,\dots,d\}$
such that their $i$-coordinates differ by exactly $2n$, and all other coordinates
are equal.  

Under Assumption~\ref{periodic_assump}, 
the process is invariant under automorphisms of $L$ and, furthermore,
the quantity
$\el\s_x\er$ does not depend on the choice of $x$.  
Let $0$ denote some fixed but arbitrary point of $K$, and let
$M=M(\l,\d,\g)=\el\s_0\er$ denote the common value of the $\el\s_x\er$.  

For $x,y\in K$, we write $x\sim y$ if $x=(u,t)$ and
$y=(v,t)$ for some $t \ge 0$ and $u,v$ adjacent in $L$.  We write 
$\{x\lrao z y\}$ for the complement of the event that there exists
an open path from $x$ to $y$ not containing $z$.
Thus, $x\lrao z y$ if: either $x \nlra y$, or $x \lra y$ and every
open path from $x$ to $y$ passes through $z$.

\begin{theorem}\label{three_ineq_lem}
Under Assumption~\ref{periodic_assump}, the following hold. 
\begin{align}\label{dg_bound_eq}
\frac{\partial M}{\partial\g}&=\frac{1}{Z^2}\int_K dx\;
E\bigl(\partial\psi_1^{0x}\partial\psi_2^\es
\cdot 1\{0\nlra \Gh\}\bigr)
\leq \frac{M}{\g}.\\
\label{dl_bound_eq}
\frac{\partial M}{\partial\l}&=\frac{1}{2Z^2}\int_K dx
\sum_{y\sim x} E\bigl(\partial\psi_1^{0xy\Gh}\partial\psi_2^\es
\cdot 1\{0\nlra \Gh\}\bigr)
\leq 2dM\frac{\partial M}{\partial\g}.\\
\label{dd_bound_eq}
-\frac{\partial M}{\partial\d}&=
\frac{2}{Z^2}\int_K dx \: E\bigl(\partial\psi_1^{0\Gh}\partial\psi_2^\es
\cdot 1\{0\overset{x}{\lra} \Gh\}\bigr)
\leq \frac{2M}{1-M^2}\frac{\partial M}{\partial\g}.
\end{align}
\end{theorem}

\begin{proof}
With the exception of~\eqref{dd_bound_eq}, the proofs mimic those 
of \cite{abf} for
the classical Ising model, and are therefore omitted. See \cite{bjo_phd}.

Here is the proof of \eqref{dd_bound_eq}. Let $|\cdot|$ denote Lebesgue
measure as usual.  By differentiating 
\begin{equation}
M=\frac{E(\partial\psi^{0\Gh})}{E(\partial\psi^\es)}=
\frac{E(\exp(2\d|\ev(\psi^{0\Gh})|))}{E(\exp(2\d|\ev(\psi^\es)|))}
\end{equation}
with respect to $\d$, we obtain that
\begin{align}
\frac{\partial M}{\partial \d}&=
\frac{2}{Z^2}E\bigl(\partial\psi_1^{0\Gh}\partial\psi_2^\es\cdot
\bigl[|\ev(\psi_1^{0\Gh})|-|\ev(\psi_2^\es)|\bigr]\bigr)\label{ihp17}\\
&=\frac{2}{Z^2}\int dx\,
E\bigl(\partial\psi_1^{0\Gh}\partial\psi_2^\es\cdot
\bigl[1\{x\in\odd(\psi_2^\es)\}-1\{x\in\odd(\psi_1^{0\Gh})\}\bigr]\bigr).
\nonumber
\end{align}

Consider the integrand in \eqref{ihp17}.
Since $\psi_2^\es$ has no sources, all odd routes in $\psi_2^\es$ are
necessarily cycles.  If $x\in\odd(\psi_2^\es)$, then $x$ lies in an odd
cycle.  We may assume that $x$ is not the endpoint of a bridge, since this
event has probability 0.  It follows that, on the event $\{0\lra \Gh\}$,
there exists an open path from $0$ to $\Gh$ that avoids $x$
(since any path can be re-routed around the odd cycle of $\psi_2^\es$ containing
$x$). 
Therefore, the event $\{0\lrao{x} \Gh\}$ does not occur, and hence
\begin{align} 
&E\bigl(\partial\psi_1^{0\Gh}\partial\psi_2^\es\cdot
1\{x\in\odd(\psi_2^\es)\}\bigr)\label{dd_pf_1_eq}\\
&\hskip2cm =E\bigl(\partial\psi_1^{0\Gh}\partial\psi_2^\es\cdot
1\{x\in\odd(\psi_2^\es)\}\cdot 
1\{0\lrao{x} \Gh\}^\tc\bigr).
\nonumber
\end{align}
  
If $\partial\psi_1^{0\Gh}\ne 0$ and
$0\lrao x\Gh$, then necessarily
$x\in\odd(\psi_1^{0\Gh})$. Hence,
\begin{align}\label{dd_pf_2_eq}
&E\bigl(\partial\psi_1^{0\Gh}\partial\psi_2^\es\cdot
1\{x\in\odd(\psi_1^{0\Gh})\}\bigr)\\
&\hskip1cm =E\bigl(\partial\psi_1^{0\Gh}\partial\psi_2^\es\cdot
1\{x\in\odd(\psi_1^{0\Gh})\}\cdot 1\{0\overset{x}{\lra} \Gh\}^\tc\bigr)\nonumber\\
&\hskip5cm +E\bigl(\partial\psi_1^{0\Gh}\partial\psi_2^\es\cdot 
1\{0\overset{x}{\lra} \Gh\}\bigr).
\nonumber
\end{align}
We wish to switch the sources $0\Gh$ from $\psi_1$ to $\psi_2$ in the right
side of \eqref{dd_pf_2_eq}.
For this we need to adapt some details of the proof of the 
switching lemma to this situation.  
The first step in the proof of that lemma
was to condition on the union $Q$ of the bridges and
ghost-bonds of the two colourings;  then, the paths from $0$ to $\Gh$ in $Q$ were listed
in a fixed \emph{but arbitrary} order. 
We are free to choose this
ordering in such a way that paths not containing $x$ have precedence, and
we assume henceforth that the ordering is thus chosen.  The next step is
to find the earliest open path $\pi$, and `add $\pi$ modulo 2' to both
$\psi_1^{0\Gh}$ and $\psi_2^\es$.  On the event 
$\{0\overset{x}{\lra} \Gh\}^\tc$, this earliest path $\pi$ does not
contain $x$, by our choice of ordering.  Hence, in the new colouring
$\psi_1^\es$,  $x$ continues to lie in an `odd' interval (recall
that, outside $\pi$, the colourings are unchanged by the switching procedure).
Therefore,
\begin{align}
&E\bigl(\partial\psi_1^{0\Gh}\partial\psi_2^\es\cdot
1\{x\in\odd(\psi_1^{0\Gh})\}\cdot 1\{0\overset{x}{\lra} \Gh\}^\tc\bigr)\\
&\hskip 2cm =E\bigl(\partial\psi_1^\es\partial\psi_2^{0\Gh}\cdot
1\{x\in\odd(\psi_1^\es)\}\cdot 1\{0\overset{x}{\lra} \Gh\}^\tc\bigr).
\nonumber
\end{align}
Relabelling, putting this into~\eqref{dd_pf_2_eq}, and
subtracting~\eqref{dd_pf_2_eq} from~\eqref{dd_pf_1_eq}, we obtain
\begin{equation}\label{dd_pf_3_eq}
\frac{\partial M}{\partial \d}=
-\frac{2}{Z^2}\int dx\:
E\bigl(\partial\psi_1^{0\Gh}\partial\psi_2^\es\cdot 
1\{0\overset{x}{\lra} \Gh\}\bigr)
\end{equation}
as required.

Turning to the inequality, let  $C^x_z$ denote the set of
points that can be reached from $z$ along open paths 
\emph{not containing $x$}.  When calculating the conditional expectation of 
$\partial\psi_1^{0\Gh}\partial\psi_2^\es\cdot 
1\{0\overset{x}{\lra} \Gh\}$ given $C^x_0$, as in the proof of the
\ghs\ inequality, we find that $\psi_1^{0\Gh}$ is a combination of two independent
colourings, one inside $C^x_0$ with sources $0x$, and one outside $C^x_0$
with sources $x\Gh$.  As in \eqref{o19}, using Lemma \ref{cor_mon_lem} as there,
\begin{align}
\label{dd_pf_4_eq}
E\bigl(\partial\psi_1^{0\Gh}\partial\psi_2^\es
\cdot 1\{0\overset{x}{\lra} \Gh\}\bigr)&=
E\bigl(\partial\psi_1^{0x}\partial\psi_2^\es\el\s_x\er_{K\setminus C^x_0}
\cdot 1\{0\overset{x}{\lra} \Gh\}\bigr)\\ 
&\leq M\cdot E\bigl(\partial\psi_1^{0x}\partial\psi_2^\es
\cdot 1\{0\overset{x}{\lra} \Gh\}\bigr).
\nonumber
\end{align}
We split the expectation on the right side according to whether or not
$x\lra \Gh$.  Clearly,
\begin{equation}\label{dd_pf_5_eq}
E\bigl(\partial\psi_1^{0x}\partial\psi_2^\es
\cdot 1\{0\overset{x}{\lra} \Gh\}
\cdot 1\{x\nlra \Gh\}\bigr)\leq
E\bigl(\partial\psi_1^{0x}\partial\psi_2^\es
\cdot 1\{x\nlra \Gh\}\bigr).
\end{equation}
By the switching lemma \ref{sl}, the other term satisfies
\begin{equation}
E\bigl(\partial\psi_1^{0x}\partial\psi_2^\es
\cdot 1\{0\overset{x}{\lra} \Gh\}
\cdot 1\{x\lra \Gh\}\bigr)=
E\bigl(\partial\psi_1^{0\Gh}\partial\psi_2^{x\Gh}
\cdot 1\{0\overset{x}{\lra} \Gh\}\bigr).
\end{equation}
We again condition on a cluster, this time $C^x_\Gh$, to obtain as in
\eqref{dd_pf_4_eq} that
\begin{equation}\label{dd_pf_6_eq}
E\bigl(\partial\psi_1^{0\Gh}\partial\psi_2^{x\Gh}
\cdot 1\{0\overset{x}{\lra} \Gh\}\bigr)\leq 
M\cdot E\bigl(\partial\psi_1^{0\Gh}\partial\psi_2^\es
\cdot 1\{0\overset{x}{\lra} \Gh\}\bigr).
\end{equation}
Combining \eqref{dd_pf_4_eq}, \eqref{dd_pf_5_eq}, \eqref{dd_pf_6_eq} with \eqref{dd_pf_3_eq}, 
we obtain by \eqref{dg_bound_eq} that 
\begin{equation}
-\frac{\partial M}{\partial \d}\leq 2M\frac{\partial M}{\partial \g}+
M^2\Big(-\frac{\partial M}{\partial \d}\Big),
\end{equation}
as required.
\end{proof}

\section{Proof of Theorem \ref{main_pdi_thm}}\label{pf_sec}

In this section we will prove the differential inequality \eqref{ihp18} which, in
combination with the inequalities of the previous section, will
yield information about the critical behaviour of the space--time Ising model.  The
proof proceeds roughly as follows.  In the
random-parity representation of $M=\el\s_0\er$, there is a backbone from $0$
to $\Gh$ (that is, to some point $g \in G$).  
We introduce two new  sourceless configurations;
depending on how the backbone interacts with these configurations, the
switching lemma allows a decomposition into a combination of
other configurations which, via Theorem \ref{three_ineq_lem}, may be 
expressed in terms of derivatives of the magnetization.

Throughout this section we work under Assumption~\ref{periodic_assump},
  that is, \emph{we work with a translation-invariant nearest-neighbour model on 
a cube in the $d$-dimensional
  lattice}, while noting that our conclusions are valid
for more general interactions with similar symmetries.
The arguments in this section borrow heavily from~\cite{abf}.
As in Theorem \ref{three_ineq_lem}, the main novelty in the proof
concerns connectivity in the `vertical' direction 
(the term $R_v$ in \eqref{ihp19}--\eqref{ihp20}
below). 

By Theorem \ref{rcr_thm},
\begin{equation}
M=\frac{1}{Z}E(\partial\psi_1^{0\Gh})
=\frac{1}{Z^3}E(\partial\psi_1^{0\Gh}\partial\psi_2^\es
\partial\psi_3^\es).
\end{equation}
We shall consider the backbone
$\xi=\xi(\psi_1^{0\Gh})$ and the open cluster $C_\Gh$ of $\Gh$ in 
$(\psi_2^\es,\psi_3^\es,\D)$.
All connectivities will refer
to the triple $(\psi_2^\es,\psi_3^\es,\D)$.  Note that $\xi$ consists of
a single path with endpoints $0$ and $\Gh$.
There are four possibilities, illustrated in Figure~\ref{decomp_fig}, for the way in which $\xi$,
viewed as a directed path from $0$ to $\Gh$, interacts with $C_\Gh$:  
\begin{romlist}
\item $\xi\cap C_\Gh$ is empty, 
\item $0 \in \xi\cap C_\Gh$, 
\item $0 \notin \xi\cap C_\Gh$, and $\xi$ first meets $C_\Gh$ immediately
  after a bridge, 
\item $0 \notin \xi\cap C_\Gh$, and $\xi$ first meets $C_\Gh$ at a cut, which
necessarily belongs to $\ev(\psi_2^\es)\cap\ev(\psi_3^\es)$.  
\end{romlist}

Thus,
\begin{equation}
M=T+R_0+R_h+R_v,
\label{ihp19}
\end{equation}
where
\begin{equation}
\label{ihp20}
\begin{split}
T&=\frac{1}{Z^3}E\bigl(\partial\psi_1^{0\Gh}\partial\psi_2^\es
\partial\psi_3^\es\cdot 1\{\xi\cap C_\Gh=\es\}\bigr),\\
R_0&=\frac{1}{Z^3}E\bigl(\partial\psi_1^{0\Gh}\partial\psi_2^\es
\partial\psi_3^\es\cdot 1\{0\lra \Gh\}\bigr),\\
R_h&=\frac{1}{Z^3}E\bigl(\partial\psi_1^{0\Gh}\partial\psi_2^\es
\partial\psi_3^\es\cdot 
1\{\mbox{first point of $\xi\cap C_\Gh$ is at a bridge of $\xi$}\}\bigr),\\
R_v&=\frac{1}{Z^3}E\bigl(\partial\psi_1^{0\Gh}\partial\psi_2^\es
\partial\psi_3^\es\cdot 
1\{\mbox{first point of $\xi\cap C_\Gh$ is a cut}\}\bigr).
\end{split}
\end{equation}
We will bound each of these terms in turn.

\begin{figure}[tbp]
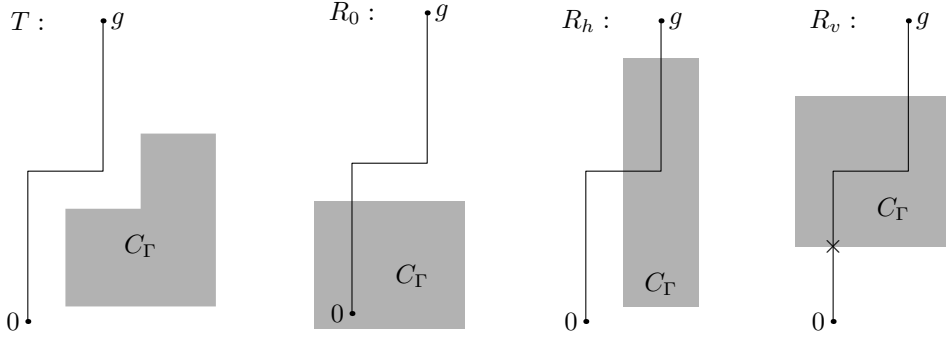

\includegraphics{sharptransition.14}
\hspace{1cm}
\includegraphics{sharptransition.15}
\hspace{1cm}
\includegraphics{sharptransition.16}
\hspace{1cm}
\includegraphics{sharptransition.17}
\caption{Illustrations of the four possibilities for $\xi\cap C_\Gh$.  
Ghost-bonds in $\psi^{0\Gh}$ are labelled $g$. The backbone $\xi$ is drawn as a
  solid black line, and $C_\Gh$ as a grey rectangle.}
\label{decomp_fig}
\end{figure}

By the switching lemma,
\begin{align}
R_0&=\frac{1}{Z^3}E\bigl(\partial\psi_1^{0\Gh}\partial\psi_2^\es
\partial\psi_3^\es\cdot 1\{0\lra \Gh\}\bigr)\label{ihp21}\\
&=\frac{1}{Z^3}E\bigl(\partial\psi_1^{0\Gh}\partial\psi_2^{0\Gh}
\partial\psi_3^{0\Gh}\bigr)=M^3.
\nonumber
\end{align}

Next, we bound $T$. The letter $\xi$ will always denote the backbone
of the first colouring $\psi_1$, with corresponding sources. 
Let $X$ denote the location of the ghost-bond that ends
$\xi$.  By conditioning on $X$, 
\begin{equation}
\begin{split}
T&=\frac{1}{Z^3}\int P(X\in dx)\,E\bigl(\partial\psi_1^{0\Gh}\partial\psi_2^\es
\partial\psi_3^\es\cdot 1\{\xi\cap C_\Gh=\es\}\bigmid X=x\bigr)\\
&\leq \frac{\g}{Z^3}\int dx\, E\bigl(\partial\psi_1^{0x}\partial\psi_2^\es
\partial\psi_3^\es\cdot 1\{\xi\cap C_\Gh=\es\}\bigr).
\end{split}
\label{o51}
\end{equation}
We study the last expectation by conditioning on $C_\Gh$ and bringing one of
the factors $1/Z$ inside.  By \eqref{backb_cond_eq}--\eqref{ihp16} and conditional
expectation, 
\begin{align}
&\frac{1}{Z}E\bigl(\partial\psi_1^{0x}\cdot 1\{\xi\cap C_\Gh=\es\}\bigmid C_\Gh\bigr)\label{o23}\\
&\hskip3cm =E\Bigl( Z^{-1}E(\pd\psi_1^{0x} \mid \xi, C_\Gh)1\{\xi\cap C_\Gh=\es\} \Bigmid C_\Gh\Bigr)
\nonumber\\
&\hskip3cm =E\bigl(\wt^{0x}(\xi)\cdot 1\{\xi\cap C_\Gh=\es\}\bigmid C_\Gh\bigr).
\nonumber
\end{align}
By Lemma \ref{rw_mon_lem},
\begin{equation}
\wt^{0x}(\xi) \le 2^{r(\xi)-r'(\xi)}
\wt_{K\setminus C_\Gh}^{0x}(\xi)\quad\mbox{on}\quad \{\xi\cap C_\Gh = \es\},
\label{o22}
\end{equation}
where
\begin{equation*}
r(\xi) = r(\xi,K),\qquad 
r'(\xi) = r(\xi,K\sm C_\Gh).
\end{equation*}
Using~\eqref{special2} and~\eqref{backbone_rep_eq}, we have
\begin{align}
&E\bigl(\wt^{0x}(\xi)
\cdot 1\{\xi\cap C_\Gh=\es\}\bigmid C_\Gh\bigr)\label{o24} \\
&\hskip3cm \leq E\bigl( 2^{r(\xi)-r'(\xi)}\wt_{K\setminus C_\Gh}^{0x} (\xi)\cdot
1\{\xi\cap C_\Gh=\es\}\bigmid C_\Gh\bigr)\nonumber \\
&\hskip3cm  \le\el\s_0\s_x\er_{K\setminus C_\Gh}.
\nonumber\end{align}

The last equation merits explanation. 
Recall that $\xi=\xi(\psi_1^{0x})$, and assume $\xi\cap C_\Gh=\es$.
Apart from the randomization that takes place when $\psi_1^{0x}$ is one of
several valid colourings,  the law of $\xi$, $P(\xi\in d\nu)$,
 is a function of the positions of
bridges and ghost-bonds along $\nu$ only, that is, the existence of
bridges where needed, and the non-existence of ghost-bonds along $\nu$.
By \eqref{o22} and Lemma~\ref{rw_mon_lem},
with $\Xi_{K\sm C} := \{\nu\in \Xi: \nu\cap C = \es\}$ and $P$
the law of $\xi$,
\begin{align*}
&E\bigl(\wt^{0x}(\xi)\cdot 1\{\xi\cap C_\Gh=\es\}\bigmid C_\Gh\bigr)\\
&\hskip3cm =\int_{\Xi_{K\sm C_\Gh}} w^{0x}(\nu)\, P(d\nu)\\
&\hskip3cm \le \int_{\Xi_{K\sm C_\Gh}}2^{r(\nu)-r'(\nu)} w^{0x}_{K\sm C_\Gh}(\nu) 
\left(\tfrac12\right)^{r(\nu)} \mu(d\nu)  
\end{align*}
for some measure $\mu$, where the factor $(\frac12)^{r(\nu)}$ arises from the
possible existence of more than one valid colouring.  
Now, $\mu$ is a measure on paths which, by the remark above, depends only locally
on $\nu$, in the sense that $\mu(d\nu)$ depends only on the bridge- and
ghost-bond configurations along $\nu$.  In particular, the same measure $\mu$
governs also the law of the backbone in the \emph{smaller} region 
$K\setminus C_\Gh$.  More explicitly, by \eqref{backbone_rep_eq} with
$P_{K\sm C_\Gh}$ the law of the backbone of the colouring
$\psi_{K\sm  C_\Gh}^{0x}$ defined on $K\sm C_\Gh$, we have
\begin{align*}
\el\s_0\s_x\er_{K\setminus C_\Gh}
&=\int_{\Xi_{K\sm C_\Gh}} w^{0x}_{K\sm C_\Gh}(\nu)\, P_{K\sm C_\Gh}(d\nu)\\
&= \int_{\Xi_{K\sm C_\Gh}} w^{0x}_{K\sm C_\Gh}(\nu) 
\left(\tfrac12\right)^{r'(\nu)}\, \mu(d\nu).
\end{align*}
Thus \eqref{o24} follows. 

Therefore, by \eqref{o51}--\eqref{o24},
\begin{align}
T&\leq \frac{\g}{Z^2}\int dx\: E\bigl(\partial\psi_2^\es
\partial\psi_3^\es\el\s_0\s_x\er_{K\setminus C_\Gh}\cdot
1\{0\nlra \Gh\}\bigr)\\
&=\g\int dx\: \frac{1}{Z^2}E\bigl(\partial\psi_2^{0x}
\partial\psi_3^\es\cdot 1\{0\nlra \Gh\}\bigr)\nonumber\\
&=\g\frac{\partial M}{\partial \g},
\nonumber
\end{align}
by `conditioning on the cluster' $C_\Gh$
and Theorem~\ref{three_ineq_lem}.

Next, we bound $R_h$.  Suppose that the bridge bringing $\xi$
into 
$C_\Gh$ has endpoints $X$ and $Y$, where we take $X$ to be the endpoint not in
$C_\Gh$.  When the bridge $XY$ is removed, the backbone $\xi$ consists of two
paths:  $\zeta^1:0\rightarrow X$ and $\zeta^2:Y\rightarrow \Gh$.  Therefore,
\begin{align*}
R_h&= \frac{1}{Z^3}\int P(X\in dx)\,E\bigl(\partial\psi_1^{0\Gh}
\partial\psi_2^\es\partial\psi_3^\es\bigmid X=x\bigr)\\
&\le\frac{\l}{Z^3}\int dx \, \sum_{y\sim x}E\bigl(\pd\psi_1^{0xy\Gh}
\partial\psi_2^\es\partial\psi_3^\es\cdot1\{0\nlra\Gh,\,y\lra\Gh\}\cdot 1\{J_\xi\}\bigr),
\nonumber
\end{align*}
where $\xi=\xi(\psi_1^{0xy\Gh})$ and
$$
J_\xi=\bigl\{\xi=\zeta^1\circ\zeta^2,
\,\zeta^1:0\rightarrow x,\,\zeta^2:y\rightarrow \Gh,\,
\zeta^1\cap C_\Gh=\es\bigr\}.
$$
As in \eqref{o23},
\begin{equation}
R_h
\leq \frac{\l}{Z^2}\int dx\:\sum_{y\sim x}
E\bigl(\partial\psi_2^\es\partial\psi_3^\es
\cdot 1\{0\nlra \Gh,\, y\lra \Gh\}\cdot
\wt^{0xy\Gh}(\xi)\cdot1\{J_\xi\}\bigr).
\label{ihp23}
\end{equation}

By Lemmas \ref{backb2}(a) and \ref{rw_mon_lem}, on the event $J_\xi$,
\begin{align*}
\wt^{0xy\Gh}(\xi) &= \wt^{0x}(\zeta^1) \wt^{y\Gh}_{K\sm \zeta^1}(\zeta^2)\\
&\le 2^{r-r'}\wt^{0x}_{K\sm C_\Gh}(\zeta^1)\wt^{y\Gh}_{K\sm \zeta^1}(\zeta^2),
\end{align*}
where $r = r(\zeta^1,K)$ and $r'= r(\zeta^1, K\sm C_\Gh)$.
By Lemma~\ref{cor_mon_lem} and the reasoning after~\eqref{o24},
\begin{align*}
E\bigl(\wt^{0xy\Gh}(\xi) \cdot 1\{J_\xi\}\bigmid \zeta^1, C_\Gh\bigr) 
&\leq 2^{r-r'}
\wt_{K\setminus C_\Gh}^{0x}(\zeta^1)\cdot \el\s_y\er_{K\sm \zeta^1}\\
&\leq M\cdot 2^{r-r'} \wt_{K\setminus C_\Gh}^{0x}(\zeta^1),
\end{align*}
so that, similarly, 
\begin{equation}
E\bigl(\wt^{0xy\Gh}(\xi)\cdot 1\{J_\xi\}\bigmid C_\Gh\bigr) \le 
M \cdot \el\s_0\s_x\er_{K\setminus C_\Gh}.
\label{o25}
\end{equation}
We substitute into the summand in \eqref{ihp23}, using the switching
lemma, conditioning on the cluster $C_\Gh$, and the bound 
$\el\s_y\er_{C_\Gh}\leq M$, to obtain the upper bound
\begin{align}
&M\cdot E\bigl(\partial\psi_2^\es\partial\psi_3^\es\cdot 
1\{0\nlra \Gh,\, y\lra \Gh\}\cdot
\el\s_0\s_x\er_{K\setminus C_\Gh}\bigr)\\
&\hskip2.5cm =M\cdot E\bigl(\partial\psi_2^{y\Gh}\partial\psi_3^{y\Gh}
\cdot 1\{0\nlra \Gh\}\cdot \el\s_0\s_x\er_{K\setminus C_\Gh}\bigr)\nonumber\\
&\hskip2.5cm =M\cdot E\bigl(\partial\psi_2^{0xy\Gh}\partial\psi_3^\es
\el\s_y\er_{C_\Gh} \cdot 1\{0\nlra \Gh\}\bigr)\nonumber\\
&\hskip2.5cm \leq M^2\cdot E\bigl(\partial\psi_2^{0xy\Gh}\partial\psi_3^\es
\cdot 1\{0\nlra \Gh\}\bigr).
\nonumber
\end{align}
Hence, by \eqref{dl_bound_eq},
\begin{align*}
R_h &\leq \l M^2\frac{1}{Z^2}\int dx\,\sum_{y\sim x}
E\bigl(\partial\psi_2^{0xy\Gh}\partial\psi_3^\es
1\{0\nlra \Gh\}\bigr)\\
&=2\l M^2 \frac{\partial M}{\partial \l}.
\end{align*}

Finally, we bound $R_v$.  Let 
$X\in\D\cap\ev(\psi_2^\es)\cap\ev(\psi_3^\es)$ be the first
point of $\xi$ in $C_\Gh$.  In a manner similar to that used for $R_h$ 
at \eqref{ihp23} above,
and by cutting the backbone $\xi$ at the point $x$,
\begin{equation}
R_v\le\frac{1}{Z^2}\int P(X\in dx)\,
E\bigl(\partial\psi_2^\es\partial\psi_3^\es
\cdot 1\{0\nlra \Gh,\,x\lra \Gh\}\cdot
\wt^{0\Gh}(\xi)\cdot 1\{J_\xi\}\bigr),
\label{ihp22}
\end{equation}
where 
$$
J_\xi=  1\bigl\{\xi=\ol\zeta^1\circ\ol\zeta^2,\, \ol\zeta^1:0\rightarrow x,\,\ol\zeta^2:x\rightarrow \Gh,
\, \zeta^1\cap C_\Gh=\es\bigr\}.
$$
As in \eqref{o25},
\begin{align*}
E(\wt^{0\Gh}(\xi)\cdot 1\{J_\xi\} \mid C_\Gh)
&= E\bigl(E(\wt^{0\Gh}(\xi)\cdot 1\{J_\xi\}\mid \ol\zeta^1,C_\Gh)\bigmid C_\Gh\bigr)\\
&\leq E\bigl(\el\s_0\s_x\er_{K\setminus C_\Gh}\cdot \el\s_x\er_{K\setminus\zeta^1}\bigmid C_\Gh\bigr)\\
&\leq \el\s_0\s_x\er_{K\setminus C_\Gh}\cdot M.
\end{align*}
By \eqref{ihp22} therefore,
\begin{equation*}
R_v\leq M\frac{1}{Z^2}\int P(X\in dx)\,
E\bigl(\partial\psi_2^\es\partial\psi_3^\es
\cdot 1\{0\nlra \Gh,\,x\lra \Gh\}
\el\s_0\s_x\er_{K\setminus C_\Gh}\bigr).
\end{equation*}
By removing the cut at $x$, the origin $0$ becomes connected to $\Gh$, but only
via $x$.  Thus,
\begin{equation*}
R_v\leq 4\d M\frac{1}{Z^2}\int dx\:
E\bigl(\partial\psi_2^\es\partial\psi_3^\es
\cdot 1\{0\overset{x}{\lra} \Gh,\,x\lra \Gh\}
\el\s_0\s_x\er_{K\setminus C^x_\Gh}\bigr),
\end{equation*}
where $C^x_\Gh$ is the set of points reached from $\Gh$ along open
paths not containing $x$.  By the switching lemma, and conditioning
twice on the cluster $C_\Gh^x$,
\begin{align*}
R_v&\leq4\d M\frac{1}{Z^2}\int dx\:
E\bigl(\partial\psi_2^{x\Gh}\partial\psi_3^{x\Gh}
\cdot 1\{0\overset{x}{\lra} \Gh\}
\el\s_0\s_x\er_{K\setminus C^x_\Gh}\bigr)\\
&=4\d M\frac{1}{Z^2}\int dx\, E\bigl(\partial\psi_2^{0\Gh}\partial\psi_3^{x\Gh}
\cdot 1\{0\overset{x}{\lra} \Gh\}\bigr)\\
&=4\d M\frac{1}{Z^2}\int dx\, E\bigl(\partial\psi_2^{0\Gh}\partial\psi_3^\es
\cdot 1\{0\overset{x}{\lra} \Gh\}\el\s_x\er_{C^x_\Gh}\bigr)\\
&\leq 4\d M^2\frac{1}{Z^2}\int dx\,
E\bigl(\partial\psi_2^{0\Gh}\partial\psi_3^\es
\cdot 1\{0\overset{x}{\lra} \Gh\}\bigr)\\
&=-2\d M^2 \frac{\partial M}{\partial \d},
\end{align*}
by \eqref{dd_bound_eq}, as required.

\section{Consequences of the inequalities}\label{cons_sec}

In this section we formulate our principal results,
and we indicate how the differential inequalities of
Theorems \ref{main_pdi_thm} and \ref{three_ineq_lem} may be used to prove them.
The arguments used are
relatively straightforward adaptations of arguments developed for the classical
Ising model, many of which are
summarized in \cite{ellis85:LD}.  In the interests of brevity, we shall omit many steps, 
and we hope that readers
familiar with the literature will be able to complete the gaps. Full details for the current
model may be found in~\cite{bjo_phd}.  
We work under Assumption~\ref{periodic_assump} throughout this section, unless
otherwise stated.  It is sometimes inconvenient to use periodic boundary
conditions, and we revert to the free condition where necessary.

We shall consider the infinite-volume limit as $L \uparrow \ZZ^d$;
the ground state is obtained by letting $\b \to\oo$ also.
Let $n$ be a positive integer, and set $L_n = [-n,n]^d$ with
periodic boundary condition. It is convenient (and equivalent)
to work instead on the translated space $\L_n^\b := [-n,n]^d \times[-\frac12 \b, \frac12\b]$,
and we assume this henceforth. By this device, the limit
process as $n,\b\to\oo$ inhabits $\ZZ^d \times \RR$ rather than $\ZZ^d \times \RRp$.
The symbol $\b$ will appear as superscript in the following; 
the superscript $\oo$ is to be interpreted
as the ground state. Let $0=(0,0)$ and
$$
M^\b_{n}(\l,\d,\g) =\el\s_0\er_{L_{n}}^\b
$$  
be the magnetization in $\L_{n}^\b$, noting that $M_n^\b\equiv 0$ when $\g=0$.

By convexity-of-pressure  arguments,
as developed in~\cite{lebowitz_martin-lof}, the limits
\begin{equation}\label{m_lim_eq}
M^\b := \lim_{n\to\oo} M^\b_n,\quad M^\oo := \lim_{n\to\oo}\lim_{\b\to\oo} M^\b_n,
\end{equation}
exist for Lebesgue-a.e.\ $\g \ge 0$.  Moreover, using the \ghs\ inequality as
in~\cite{preston_ghs} (which implies the differentiability of the pressure
function in $\g$ whenever $\g>0$) and the results
of~\cite{lebowitz_martin-lof}, we find that the
limits \eqref{m_lim_eq} exist for all $\g >0$, and are independent of
the order of the limits.  Note that this argument does not
rely on a Lee--Yang theorem.  We have that $M^\b(\l,\d,0)=0$.  

By a standard re-scaling argument, $M^\oo$
depends only on the ratios $\l/\d$ and $\g/\d$,
and thus we shall set $\d=1$, $\rho=\l/\d$,
and write
$$
M^\b(\rho,\g)= M^\b(\rho,1,\g),\qquad \b\in(0,\oo],
$$
with a similar notation
for other functions.  

As in \cite{lebowitz_martin-lof}, when $\g > 0$, there 
exists a unique equilibrium state at $(\rho,\g)$.
That is, the limits 
$$
\el\s_A\er^\b:=\lim_{n\to\oo}\el\s_A\er^\b_n, \quad \el\s_A\er^\oo:=\lim_{n,\b\to\oo}\el\s_A\er^\b_n,
$$
exist for all $A$, where $\el\cdot\er_n := \el\cdot\er_{L_n}$, and the limits are independent
of the choice of boundary condition.
It follows that the infinite-volume probability measure exists
(this is a standard exercise using the Skorohod topology, see 
\cite{bezuidenhout_grimmett,ethier_kurtz}).
A phase transition is manifested by non-uniqueness of the equilibrium state,
and this can therefore occur only when $\g=0$. Let $\el\cdot \er^\b_+$ be
the limiting state of $\el\cdot\er^\b$ as $\g \downarrow 0$, and
$$
M^\b_+(\rho):=\lim_{\g\downarrow 0}M^\b(\rho,\g).
$$
As in \cite{lebowitz_martin-lof}, there is non-uniqueness
at $(\rho,0)$ if and only if $M^\b_+(\rho)>0$, and this
motivates the definition
\begin{equation}
\bc^\b:=\inf\{\rho>0:M^\b_+(\rho)>0\},
\label{crit_val_defs_eq}
\end{equation}
see also \eqref{o1} and \eqref{critvals}. We shall have need later for
the infinite-volume limit $\el\cdot\er^{\rf,\b}$,
as $n\to\oo$, with \emph{free} boundary condition in the
$\ZZ^d$ direction.
Note that 
\begin{equation}
\el\cdot\er^{\rf,\b}_{\g=0} = \el\cdot\er^\b_{\g=0}=\el\cdot \er^\b_+ \quad \text{ if } \quad M^\b_+(\rho)=0.
\label{o40}
\end{equation}
The superscript `f' shall always indicate this free boundary condition.

\begin{remark}\label{rc_unique}
It is sometimes convenient to work with the \rc\ (or \fk) representation
of the space--time Ising model, as in \cite{akn,grimmett_stp,GOS}. For $\b\in(0,\oo)$,
let $\fr^{b,\b}$, $b=0,1$, be the $q=2$ \rc\ measures arising
as the limit as $n\to\oo$ of the continuum \rc\ measure on $L_n\times[-\frac12\b,\frac12\b]$
with respectively free/wired boundary condition in the spatial direction.
(There are no ghost-bonds, in that $\g=0$.)
We define $\fr^{b,\oo}$ similarly. 
As discussed in \cite{akn,GOS}, and in \cite{grimmett_RCM} for discrete lattices,
these limits exist, and are equal for all but countably many
values of $\rho$. (They are presumably equal for all
$\rho\ne\bc$, using arguments of \cite{ACCN,bod06,grimmett_RCM}, 
but we do not pursue this further here.)
Furthermore, they are non-decreasing
in $\rho$, and, in particular, 
\begin{equation}
\label{mel61}
\fr^{1,\b} \le \frr^{0,\b}, \qquad \rho<\rho',
\end{equation}
where $\le$ denotes stochastic ordering (see \cite{GOS}). 
In the usual manner, for $\b\in(0,\oo]$,
\begin{equation}
\fr^{1,\b}(x\lra y) = \el\s_x\s_y\er^\b_+, \quad \fr^{1,\b}(0\lra\oo)
=M_+(\rho),
\label{mel60}
\end{equation}
where $\lra$ denotes an open connection in the \rc\ model.
It may be seen as in  
\cite[Thms 4.19, 4.23]{grimmett_RCM} that the $\fr^{b,\b}$ 
have trivial tail $\s$-fields, and are thus mixing and ergodic.
Therefore, the $\fr^{b,\b}$
possess (a.s.) no more than one unbounded cluster, by the 
Burton--Keane argument, \cite{burton_keane,grimmett_RCM}.
By \eqref{mel60}, the
{\fkg} inequality, and the uniqueness of any unbounded cluster,
\begin{equation}
\el\s_x \s_y\er^\b_+ \ge \fr^{1,\b}(x\lra\oo)\fr^{1,\b}(y\lra\oo)
= M^\b_+(\rho)^2.
\label{o50}
\end{equation}
\end{remark}

Let $\b\in(0,\oo)$. Using the convexity of Lemma \ref{ghs_lem}
as in \cite{ellis85:LD}, the derivative $\pd M^\b/\pd\g$ exists for almost every $\g\in(0,\infty)$,
and, when this holds,
\begin{equation}\label{chi_lim_eq}
\chi^\b_n(\rho,\g):=\frac{\partial M^\b_n}{\partial\g}\rightarrow
\chi(\rho,\g):=\frac{\partial M^\b}{\partial\g}< \infty.
\end{equation}
The corresponding conclusion holds also as $n,\b\to\oo$.  
Furthermore, the limits
$$
\chi^\b_+(\rho) := \lim_{\g\downarrow0} \chi^\b(\rho,\g), \qquad \b\in(0,\oo],
$$
exist when taken along suitable sequences. 

The limit
\begin{align}
\chi^{\rf,\b}(\rho,0) &:= \lim_{n\to\oo}\left(\left.\frac{\pd M^{\rf,\b}_n}{\pd \gamma}\right|_{\g=0}\right)
\label{mel65}\\
&=\lim_{n\to\oo}\int_{\L_n^\b} \el\s_0\s_x\er_{n,\g=0}^{\rf,\b}\,dx
=\int\el\s_0\s_x\er^{\rf,\b}_{\g=0}\,dx
\nonumber
\end{align} 
exists by monotone convergence, see Lemma \ref{cor_mon_lem}.
By Lemma \ref{ghs_lem}, 
\begin{equation}
\chi_+^\b(\rho) \ge \chi^{\rf,\b}(\rho,0)\ \text{whenever}\ M_+^\b(\rho)=0,
\qquad \b\in(0,\oo].
\label{mel63}
\end{equation}
Let
\begin{equation}
\bs^\b :=\inf\{\rho>0:\chi^{\rf,\b}(\rho,0)=\infty\},\qquad \b\in(0,\oo].
\label{o27}
\end{equation}
By \eqref{mel61}--\eqref{mel60} and the monotonicity of $\chi^{\rf,\b}(\rho,0)$, 
\begin{equation}
\bs^\b \le \bc^\b.
\label{mel64}
\end{equation}
By the discussion around \eqref{crit_val_defs_eq}--\eqref{o40}, 
there is a unique equilibrium state when $\g=0$ and $\rho<\bc^\b$.  
We shall see in Theorem \ref{0mass} that $\chi^{\rf,\b}(\bs^\b,0)=\oo$.

For $x\in\ZZ^d\times\RR$, let $\|x\|$ denote the supremum norm of $x$.

\begin{theorem}\label{exp_decay_cor}
Let $\b\in(0,\oo]$ and $\rho <\bs^\b$.  There exists $\a=\a^\b(\rho)>0$ such that
\begin{equation}
\el\s_0\s_x\er^\b_+ \leq e^{-\a\|x\|},\qquad x\in\ZZ^d\times\RR.
\end{equation}
\end{theorem}

\begin{proof}
Fix $\b\in(0,\oo)$ and $\g=0$, and let $\rho<\bs^\b$, so that \eqref{o40} applies. 
Therefore,
\begin{equation}
\chi^{\rf,\b}(\rho,0)=\int_{\ZZ^d\times[-\frac12\b,\frac12\b]}\el\s_0\s_x\er^\b \;dx=
\sum_{k\geq1}\int_{C_k^\b}\el\s_0\s_x\er^\b \;dx,
\end{equation}
where $C_k^\b:=\L^\b_k\setminus\L^\b_{k-1}$.
Since $\rho<\bs^\b$, the last 
summation converges, whence, for sufficiently large $k$,
\begin{equation}\label{exp_cond_eq}
\int_{C_k^\b}\el\s_0\s_x\er^\b \, dx<e^{-8}.
\end{equation}
The result now follows in the usual manner
by the Simon inequality, Lemma~\ref{simon_lem}, with the 1-fat separating sets $C_k^\b$.
A similar argument holds when $\b=\oo$.
Further discussion of the method may be found at \cite[Corollary~9.38]{grimmett_RCM}. 
\end{proof}

Let $\b\in(0,\oo]$, $\g=0$ and define the \emph{mass}
\begin{equation}
m^\b(\rho):=\liminf_{|x|\rightarrow\infty}
\left(-\frac{1}{\|x\|}\log\el\s_0\s_x\er^\b_\rho\right)
\end{equation}
By Theorem \ref{exp_decay_cor} and \eqref{o50},
\begin{equation}
m^\b(\rho) 
\begin{cases} >0 &\text{if  }\rho<\bs^\b,\\
=0 &\text{if } \rho>\bc^\b.
\end{cases}
\label{mel66}
\end{equation}

\begin{theorem}\label{0mass}
Except when $d=1$ and $\b<\oo$, $m^\b(\bs^\b)=0$ and $\chi^{\rf,\b}(\bs^\b,0)=\oo$. 
\end{theorem}

\begin{remark}\label{Lebowitz-ineq}
The manner of the divergence of the susceptibility $\chi$ may be studied via the so-called
Lebowitz inequalities of \cite{leb74}. Such
inequalities are easily proved for the quantum Ising model using the switching lemma.
\end{remark}

\begin{proof}
Let $d \ge 2$, $\g=0$, and fix $\b\in(0,\oo)$.
We use the Lieb inequality,
Lemma~\ref{lieb_lem}, and the argument of~\cite{lieb80,simon80},
see also \cite[Corollary 9.46]{grimmett_RCM}.
It is necessary and sufficient for $m^\b(\rho)>0$ that
\begin{equation}\label{exp_cond_eq_2}
\int_{C_n^\b}\el\s_0\s_x\er^{\rf,\b}_{n,\rho}\, dx<e^{-8}\quad\mbox{for some }n.
\end{equation}
Necessity holds because the integrand is no greater than $\el\s_0\s_x\er^\b$.
Sufficiency follows from Lemma \ref{lieb_lem}, as in the proof of Theorem \ref{exp_decay_cor}.

By \eqref{st_Ising_eq},
\begin{align*}
\frac{\partial}{\partial\rho}\el\s_0\s_x\er^{\rf,\b}_{n,\rho}
&=
\tfrac12 \int_{\L_n^\b}dy\,\sum_{z\sim y}\el\s_0\s_x;\s_y\s_z\er^{\rf,\b}_{n,\rho}\\
&\leq d\b(2n+1)^d.
\end{align*}
Therefore, if $\rho'>\rho$,
\begin{equation}
\int_{C_n^\b}\el\s_0\s_x\er^{\rf,\b}_{n,\rho'}\,dx\leq d[\b(2n+1)^d]^2(\rho'-\rho)
+\int_{C_n^\b}\el\s_0\s_x\er^{\rf,\b}_{n,\rho}\,dx.
\end{equation}
Hence, if \eqref{exp_cond_eq_2} holds for some $\rho$, then it holds for $\rho'$
when $\rho'-\rho>0$ is sufficiently
small.  

Suppose $m^\b(\bs^\b)>0$. Then $m^\b(\rho')>0$ for some $\rho'>\bs^\b$,
which contradicts $\chi^{\rf,\b}(\rho',0)=\oo$, and the first claim of the theorem follows.
A similar argument holds when $d\ge 1$ and $\b=\oo$.
The second claim follows similarly: if $\chi^{\rf,\b}(\bs^\b,0)<\oo$, then
\eqref{exp_cond_eq_2} holds with $\rho=\bs^\b$, whence $m^\b(\rho')>0$
and $\chi^{\rf,\b}(\rho',0)<\oo$
for some $\rho'>\bs^\b$, a contradiction. (See also \cite{aizenman_tree-graph}.)
\end{proof}

We are now ready to state the main results.
The inequalities of Theorems \ref{three_ineq_lem} and \ref{main_pdi_thm}
may be combined to obtain
\begin{equation}
M^\b_n\leq (M^\b_n)^3+\chi^\b_n\cdot
\left(\g+4d\l (M^\b_n)^3+4\d \frac{(M^\b_n)^3}{1-(M^\b_n)^2}\right).
\end{equation}
Using these inequalities and the facts stated above, it is straightforward to adapt the
arguments of~\cite[Lemmas~4.1, 5.1]{ab} (see also \cite{abf,grimmett_perc})
to prove the following. We omit the proofs.

\begin{theorem}\label{ab_thm}
There are constants $c_1$, $c_2>0$ such that, for $\b\in(0,\oo]$, 
\begin{align}\label{ab_1_eq}
M^\b(\bs,\g)&\geq c_1\g^{1/3},\\
\label{ab_2_eq}
M^\b_+(\rho)&\geq c_2(\rho-\bs^\b)^{1/2},
\end{align}
for small positive $\g$ and $\rho-\bs^\b$, \resp. 
\end{theorem}

This is vacuous when $d=1$ and $\b<\oo$; see \eqref{critvals}.
The exponents in the above inequalities are presumably sharp
in the corresponding mean-field model (see \cite{abf,af} and Remark \ref{remark_mf}).
It is standard that a number of important results follow from
Theorem~\ref{ab_thm}, some of which we state here.

\begin{theorem}\label{eq_cor}
For $d \ge 1$ and $\b\in(0,\oo]$, we have that $\bc^\b=\bs^\b$.
\end{theorem}

\begin{proof}
Except when $d=1$ and $\b<\oo$, this is immediate 
from \eqref{mel64} and \eqref{ab_2_eq}.
In the remaining case, $\bc^\b=\bs^\b=\oo$.
\end{proof}

\begin{remark}\label{remark_mf}
Let $\b\in(0,\oo]$. Except when $d=1$ and $\b<\oo$,
one may conjecture the existence of exponents $a=a^\b(d)$, $b=b^\b(d)$ such that
\begin{alignat}{2}
M^\b_+(\rho)&=(\rho-\bc^\b)^{(1+\o(1))a }\qquad&&\mbox{as }\rho\downarrow\bc^\b,\\
M^\b(\bc^\b,\g)&=\g^{(1+\o(1))/b}\qquad&&\mbox{as }\g\downarrow 0.
\end{alignat}
(We do not exclude the possibility that, when $\b<\oo$, the values of the exponents 
depend also on the value of $\d$.)
Theorem~\ref{ab_thm} would then imply that $a\leq \frac12$ and $b \geq 3$.
In \cite[Thm 3.2]{chayes_ioffe_curie-weiss} it is proved for the ground-state 
quantum Curie--Weiss, or mean-field, model that the corresponding
$a =\frac12$.  It may be conjectured (as proved for the classical Ising model
in \cite{af}) 
that the values $a=\frac12$ and
$b=3$ are attained for the space--time Ising model on $\ZZ^d\times [-\frac12 \b,\frac12\b]$
for $d$ sufficiently large, that is, when either $\b<\oo$ and $d \ge 4$,
or $\b=\oo$ and $d \ge 3$.
\end{remark}

Finally, a note about \eqref{o5}. The \rc\ measure corresponding to
the quantum Ising model is \emph{periodic} in
both $\ZZ^d$ and $\b$ directions, and this complicates the infinite-volume limit.
Since the periodic \rc\ measure dominates the free \rc\ measure, for $\b\in(0,\oo)$,
as in \eqref{mel61} and \eqref{o50},
\begin{alignat*}{2}
\liminf_{n\to\oo} \tau^\b_{L_n}(u,v) &\ge \el\s_{(u,0)}\s_{(v,0)}\er_{+,\rho'}^\b \qquad&&\text{for } \rho'<\rho\\
&\to M_+^\b(\rho-)^2 \qquad &&\text{as } \rho'\uparrow \rho,
\end{alignat*}
and a similar argument holds in the ground state also.

\section{In one dimension}\label{sec_1d}

The space--time version of
the quantum Ising model on $\ZZ$ is two-dimensional,
 living on
$\ZZ\times\RR$.  In the light of \eqref{critvals}, we shall study only
the ground state, and we shall suppress the superscript
$\oo$. One may adapt some of the special
arguments for two-dimensional models based on planar duality.
One consequence is the following.

\begin{theorem}\label{crit_val_cor}
Let $d=1$.  Then $\bc=2$, and the transition is
of second order in that $M_+(2)=0$.
\end{theorem}

We mention two applications of this theorem.
Consider first a `star-like' graph, comprising finitely many
copies of $\ZZ$, pairs of which may intersect at single points.  
It is shown in \cite{bjo0}, using Theorem
\ref{crit_val_cor}, that
the quantum Ising model on such a graph has critical value $\bc = 2$.

Secondly, in an account \cite{GOS} of so-called `entanglement' in the 
quantum Ising model on the subset $[-m,m]$ of $\ZZ$,
it was shown that the reduced density matrix $\nu_{m}^L$
of the block $[-L,L]$ satisfies
$$
\|\nu_m^L - \nu_n^L\| \le \min\{2, C L^{\a}e^{-cm}\},\qquad 2\le m<n<\oo,
$$ 
where $C$ and $\a$ are constants depending on $\rho=\l/\d$,
and $c=c(\rho)>0$ whenever $\rho < 1$. Using Theorems \ref{exp_decay_cor}
and \ref{crit_val_cor}, we have that $c(\rho)>0$ if and only if $\rho< \bc = 2$. 

\begin{proof}
We sketch the proof here.
It uses the random-cluster (or \fk) representation of the
equilibrium state $\el\cdot\er_+$, see Remark \ref{rc_unique}.
Writing $\fr^0$ (\resp, $\fr^1$) for the free (\resp, wired) $q=2$ \rc\ measure, we have 
as in \eqref{mel60} that
\begin{equation}
\el\s_x\s_y\er_+=\fr^1(x\lra y),\quad\el\s_x\er_+=\fr^1(x\lra\infty).
\label{o26}
\end{equation}

Planar duality is a standard tool in two-dimensional models,
and it applies to the \rc\ model on $\ZZ\times \RR$.
The details are similar to those
in related systems, and the reader is referred to
\cite{aizenman_nacht,grimmett_RCM,grimmett_stp} in this regard.
There is a standard computation that shows that,
in a certain sense that is sensitive to
the geometry of the configurations, $\phi_\rho^0$ and $\phi_{4/\rho}^1$
form a dual pair of measures.

The argument developed by Zhang for percolation (see \cite{grimmett_perc,grimmett_RCM})
may be adapted to the current setting to obtain that $\bc \ge 2$.
Roughly speaking, this is as follows. Suppose that $\bc < 2$, so that
there exists, $\phi_2^0$-almost-surely, an unbounded cluster. As in
Remark \ref{rc_unique}, for $b=0,1$, there exists, $\phi_2^b$-almost-surely,
 a unique unbounded cluster.
This implies that both the primal and dual processes at $\rho=2$ contain
unbounded clusters, a possibility that Zhang's construction shows to be contradictory.
The argument so far uses no facts proved in the current paper, and it yields
that 
\begin{equation}
\label{o52}
\phi_2^0(0\lra \oo) = 0.
\end{equation}

We show next that $\bc\le 2$, following the method developed for
percolation to be found in \cite{grimmett_perc,grimmett_RCM}. Suppose that $\bc>2$.  
By the above duality, 
one may find a box of side-length $n$ such that: the $\phi_2^1$-probability
of a crossing of this box is bounded away from $0$ uniformly in $n$.
By \eqref{o26} and Theorem \ref{exp_decay_cor}, this probability
decays to zero in the manner of $C n e^{-\a n}$ as $n\to\oo$, a contradiction.

We show finally that $M_+(2)=0$ by adapting a simple argument 
presented by Werner in \cite{WW08} for the classical Ising model on $\ZZ^2$.
Certain geometrical details are omitted.
Let $\pi_2$ be the Ising state
obtained from a realization of $\phi_2^0$ 
by labelling each open cluster $+1$ with probability
$\frac12$, and otherwise $-1$. By \eqref{o52} and a standard argument
based on the coupling with the \rc\ measure $\phi^0_2$ 
(see \cite[Ex.\ 8.14]{G-pgs}), $\pi_2$ is ergodic.
The Ising state $\pi_2^+$ is obtained similarly from the \rc\ measure
$\phi_2^1$, with the difference that any infinite cluster is invariably
assigned spin $+1$.

We adopt the harmless
convention that, for any spin-configuration $\s$ on $\ZZ\times \RR$, 
the subset labelled $+1$ is closed; the labelling is well-defined except 
at deaths, and we choose to label a death $a$ with the spin $+1$ if and only if at least one of
the intervals abutting $a$ is labelled $+1$. 

Let $\s$ be a spin-configuration on $\ZZ\times\RR$.
The binary relations $\lrao \pm$ are defined as follows.
A \emph{path} of $\ZZ\times\RR$ is a self-avoiding path of $\RR^2$ that:
traverses a finite number of line-segments of $\ZZ\times\RR$, 
and is permitted to connect them by passing between any two points of
the form $(u,t)$, $(u\pm 1,t)$. A path is called a $(+)$path (\resp, $(-)$path)
if all its elements are labelled $+1$ (\resp, $-1$).
For $x,y\in\ZZ\times\RR$, we write 
$x\lrao + y$ (\resp, $x\lrao - y$) if there exists a $(+)$path 
(\resp, $(-)$path) with endpoints $x$, $y$.
Let $N^+$ (\resp, $N^-$) be the number of unbounded $+$ (\resp, $-$)
Ising clusters with connectivity relation $\lrao +$ (\resp, $\lrao -$).
 By the Burton--Keane argument, either
$\pi_2(N^+=1) = 1$ or $\pi_2(N^+=0)=1$. The former entails
also that $\pi_2(N^-=1)=1$, and this is impossible by another use of Zhang's argument.
Therefore, 
\begin{equation}
\pi_2(N^\pm = 0) = 1.
\label{o31} 
\end{equation}

There is a standard argument for deducing $\pi_2=\pi_2^+$ from \eqref{o31},
of which the idea is roughly as follows.
(See
\cite{ACCN} or \cite[Thm 5.33]{grimmett_RCM} for examples of similar arguments
applied to the \rc\ model.) 
Let $\L_n=[-n,n]^2$, viewed as a subset of $\ZZ\times \RR$. The \emph{boundary}
$\pd \L_n$ is defined in the usual way as the intersection of $\L_n$ with
the subset $\RR^2 \sm (-n,n)^2$ of $\RR^2$.
By \eqref{o31}, for given $m$, and for $\eps>0$ and sufficiently large $n$,
the event $A_{m,n}= \{\L_{m+1} \lrao -\pd\L_n\}^\tc$ satisfies $\pi_2(A_{m,n}) > 1-\eps$.

Let $M_n$ be the subset of $\L_n$ containing all points connected to $\pd\L_n$ by $(-)$paths of
$\L_n$. Thus $M_n$ is a union of maximal intervals, and each endpoint of such
an interval either lies in $\pd\L_n$ (and is labelled $-1$), 
or lies in $\L_n\sm\pd\L_n$ (and is labelled $+1$). 
Let $\D M_n$ be the set of all points $(u,t)\in \ZZ\times \RR$ of $\L_n\sm M_n$ satisfying: 
either (i) $(u,t)\notin \pd\L_n$ and $(u,t)$ is an endpoint of
a maximal interval of $M_n$, or (ii) there exists $e\in\{-1,+1\}$ such that $(u,t+e)
\in M_n$. By the definition of $M_n$, every point in $\D M_n$ is labelled $+1$.   

Let $m<n$, and let $I_n$ be the set of all points in $\L_n$ reachable from $\L_m$ along paths 
of $\L_n \sm \D M_n$. The random set $I_n$ is given in terms of $M_n$, and therefore
$I_n$ is measurable on the spin configuration of its complement $\L_n\sm I_n$. 
Given $I_n$, the spin configuration on $I_n$ is a space--time Ising model
with $+$ boundary conditions.
By the \fkg\ inequality, 
conditional on $I_n$ (and the event $A_{m,n}$), the conditional $\pi_2$-measure
on $\L_m$ is stochastically greater than $\pi^+_2$.
By passing to a limit, we obtain that $\pi_2\ge\pi_2^+$. Since $\pi_2\le \pi_2^+$
by elementary considerations of \fkg\ type, we deduce that $\pi_2=\pi_2^+$
as claimed.

One way to conclude that $M_+(2)=0$ is to use
the \rc\ representation again.
By \eqref{o52} and the above, 
$$
\phi_2^0(0\lra \oo) = \phi_2^1(0\lra\oo)=0,
$$
whence
$M_+(2) \le \phi_2^1(0\lra \oo) = 0$.
\end{proof}

\section*{Acknowledgements}
JEB acknowledges financial support from the Royal Institute
of Technology (Sweden), Riddarhuset, Stockholm, and the Engineering
and Physical Sciences Research Council during his
PhD studentship at the University of Cambridge.
GRG thanks the Institut Henri Poincar\'e--Centre Emile Borel, Paris, for its hospitality during
the completion of this project.

\bibliographystyle{amsplain}
\bibliography{qim}

\end{document}